\documentclass[10pt,journal]{IEEEtran}
\usepackage[justification=centering]{caption}
\usepackage{lipsum,multicol}
\usepackage{bigdelim}
\usepackage{multirow}
\usepackage{calc}
\usepackage{enumerate}
\usepackage{cite}
\usepackage{float}
\usepackage{enumerate}
\usepackage{algorithmic}
\usepackage{url}
\usepackage{amstext,amsthm,latexsym,amssymb,graphicx,tabls}
\usepackage{amssymb,amsthm}
\usepackage{amsfonts,amsmath}
\usepackage{latexsym}
\usepackage[mathscr]{euscript}
\usepackage{graphics}
\usepackage{graphicx}
\usepackage{cite}
\usepackage{mciteplus}
\usepackage{cases}
\usepackage{tabularx,bbm}
\usepackage{color}
\usepackage{epstopdf}

\newcommand{\beq}{\begin{equation}}
\newcommand{\eeq}{\end{equation}}
\newcommand{\beqq}{\begin{equation*}}
\newcommand{\eeqq}{\end{equation*}}
\newcommand{\ei}{\end{itemize}}
\newcommand{\bi}{\begin{itemize}}
\newcommand{\ee}{\end{enumerate}}
\newcommand{\be}{\begin{enumerate}}

\newtheorem{definition}{Definition}
\newtheorem{lemma}{Lemma}
\newtheorem{prop}{Proposition}

\newtheorem{corol}{Corollary}

\theoremstyle{remark}

\makeatletter
\newif\if@restonecol
\makeatother

\usepackage[ruled,vlined]{algorithm2e}

\newcommand{\ls}[1]
 {\dimen0=\fontdimen6\the\font \lineskip=#1\dimen0
\advance\lineskip.5\fontdimen5\the\font \advance\lineskip-\dimen0
\lineskiplimit=.9\lineskip \baselineskip=\lineskip
\advance\baselineskip\dimen0 \normallineskip\lineskip
\normallineskiplimit\lineskiplimit \normalbaselineskip\baselineskip
\ignorespaces }
\ls{.4}

\begin{document}
\title{NEWCAST: Anticipating Resource Management and QoE Provisioning for Mobile Video Streaming}

\author{\IEEEauthorblockN{Imen Triki,
Rachid El-Azouzi and Majed Haddad}\\
\IEEEauthorblockA {LIA/CERI, University of Avignon,\\
 Agroparc, BP 1228, 84911, Avignon, France\\
}}

\IEEEtitleabstractindextext{%
\begin{abstract}
The knowledge of future throughput variations in mobile networks becomes more and more possible today thanks to the rich contextual information provided by mobile applications and services and smartphone sensors. It is even likely that such contextual information, which may include traffic, mobility and radio conditions will lead to a novel agile resource management not yet thought of. In this paper, we propose an framework (called NEWCAST) that anticipates the throughput variations to deliver video streaming content. We develop an optimization problem that realizes a fundamental trade-off among critical metrics that impact the user's perceptual quality of experience (QoE) and the cost of system utilization. Both simulated and real-world throughput traces collected from \cite{dataset} were carried out to evaluate the performance of NEWCAST. In particular, we show from our numerical results that NEWCAST provides the efficiency that the new 5G architectures require in terms of computational complexity and robustness. We also implement a prototype system of NEWCAST and evaluate it in a real environment with a real player to show its efficiency and scalability compared to baseline adaptive bitrate algorithms.

 \end{abstract}
\vspace{0.5cm}
}

\maketitle
\IEEEdisplaynontitleabstractindextext
\IEEEpeerreviewmaketitle

\section{Introduction}
Due to the breakthrough evolution of smartphones and their large penetration in daily life, mobile networks have witnessed an unrivaled growth of their mobile traffic posing new challenges to their resource management. The evolution of multimedia services in the Internet and the increasing consumer demand for high definition (HD) contents have even led the operators and the industry to rethink the way networks are dimensioned.
According to recent statistics carried out by Cisco \cite{Cisco16}, $82\%$ of all internet consumers' traffic will be http video streaming by $2021$, which explains the huge amount of care being accorded to video streaming services. 

In the literature, many studies were carried out to identify the critical metrics that may impact the user's perceptual QoE \cite{conf/im/VriendtVR13}. One of the key factors that may reflect the users' experience is the user engagement. Authors in \cite{DOB11} quantified the user engagement and identified some critical metrics that may affect it such as the buffering ratio, the rate of buffering, the start-up delay, the rendering quality and the average bitrate. Recent works in \cite{DOB11}\cite{Sigcomm13} developed approaches to understand how some quality metrics may influence the user engagement. It was revealed through \cite{Sigcomm13} that the rebuffering events have a significant impact on the QoE in the sense that the time spent on rebuffering during a video session can significantly reduce the user engagement. One other aspect that may impact the user engagement is the temporal variations of the video quality. Indeed, authors in \cite{Yim11} claimed that temporal variability in quality can be considered as worse as a constant quality with a lower average bitrate. Additional empirical results in \cite{quality} showed that humans appear to be more forgiving on buffer stalls than they are on video quality variations. Long buffer freezing events are even not rated worse than short buffer freezing towards high video quality levels. 

To improve the user engagement in real time, DASH (Dynamic Adaptive Streaming over HTTP) appeared as an emerging standard for video content delivery. Various commercial solutions adopting DASH have been proposed to improve the user's QoE such as Microsoft's smooth streaming, Adobe's HTTP dynamic streaming and Apple's live streaming. In DASH, each video file is divided into multiple small segments encoded at multiple quality levels \cite{Stockhammer:2011:DAS:1943552.1943572}, and it is up to the client to chose the most suitable quality level (bitrate) to stream the future segment.
In the literature, adaptive bitrate algorithms are classified in three main classes: buffer-based \cite{Huang14}, throughput-based \cite{Tian:2012:TAS:2413176.2413190} and buffer--throughput-based algorithms \cite{Yin:2015:CAD:2829988.2787486}. While the first class makes the decision based on the playback buffer occupancy state, the second class exploits the historical TCP throughput measurements \cite{TCPthroughput} to estimate the current bandwidth and instantaneously adapt the quality.

Within these classes, many adaptive strategies were proposed to reduce the interruption of the playback buffer \cite{Huang14, Bouaziz}. In \cite{Yin:2015:CAD:2829988.2787486}, authors proposed a predictive control algorithm that combines throughput and buffer occupancy information. \cite{Jiang:2012:IFE:2413176.2413189} developed a suite of techniques that guide the trade-offs between stability, fairness and efficiency leading to a general framework for robust video adaptation. In \cite{journals/corr/JosephV13}, authors were addressing the resource management issue in DASH QoE provisioning while considering user preferences on rebuffering and cost of video delivery.

Although there is a rich literature on methods used for optimizing the QoE in video streaming services, very few papers were exploiting the knowledge of future throughput variations for quality adaptation. The main idea of this paper is inspired from \cite{conf/infocom/LuV13} where authors designed a QoE-driven optimization framework that exploits the knowledge of future throughput variations to minimize the system utilization cost while avoiding rebuffering events. The main shortcoming of their approach is that it is only suited for classical video streaming as it ignores important visual quality metrics related to adaptive streaming.

Recent studies on contextual information have revealed a promising possibility of accurately predicting the future available resources over a medium horizon. For instance, context acquisition can target the monitoring of contextual situations as soon as they are created. The output can describe the contexts encountered as well as the likelihood of encountering similar contexts in the future \cite{5G}. This rises the opportunity to efficiently design the bitrate adaptation by exploiting the knowledge of future capacity variations \cite{Yim11, Zou:2015:API}.

Since video streaming is very bandwidth consuming, its delivery cost became too high for operators to support the increasing bandwidth demand with the arrival of ultra high definition (UHD) video quality, which requires $16$ times more pixels than full HD. However, it is important to develop solutions taking into account the delivery cost as well as the QoE through different metrics like rebuffering, average quality and switching in quality levels. In this paper, we design a QoE-driven optimization framework that realizes the trade-off between bandwidth utilization cost and content resolution under constraints on rebuffering events. It extends the model developed in \cite{conf/infocom/LuV13} by considering adaptive video streaming.
We summarize our main contributions as follows:
\begin{itemize}
\item We provide a general optimization framework for stored video delivery that accounts for heterogeneous client preferences, QoE models and capacity variations,
\item Under the constraint of no rebuffering events, we formally obtain an optimal solution where the transmission schedule is of a threshold type and the bitrate distribution is of an ascending order,
\item We propose an efficient heuristic, which we call NEWCAST, that performs close to the optimal approach. NEWCAST's performances are evaluated through simulations under the constraint of no rebuffering events, then under the hypothesis of tolerating buffer stall during the streaming session,
\item We study the characteristics of NEWCAST in terms of robustness (using real traces) and complexity. We then compare it to baseline adaptive bitrate algorithms,
\item We implement NEWCAST in a real environment and adapt it for real interactions with a real DASH player.

\end{itemize}

The rest of the paper is organized as follows: In Section \ref{sec:Assumptions}, we introduce the system model and formulate the optimization problem. In Section \ref{sec:Resoluion}, we discuss the properties of the optimal solution. In Section \ref{algo}, we propose optimal approaches and heuristic algorithms for the problem resolution with the constraint of no rebuffering events. Then, in Section \ref{algo2}, we consider the hyposesis to allow rebuffering events during the streaming session. Section \ref{sec:results} is dedicated to both simulations and numerical results and Section \ref{sec:experiments} is dedicated to experiments. We conclude the paper in Section \ref{sec:conc}.

\section{Problem formulation}\label{sec:Assumptions}
We consider a video file stored in a video streaming server and divided into $N$ segments of equal length in second. Each segment is composed of $S$ frames and encoded at $L$ different bitrates $\{b_1, \ldots, b_L\}$, such that $b_i < b_j$ for $i<j$. To stream the video, the client requests the segments to the server one by one and indicates at each request the video quality (bitrate) needed for the streaming. Denote by $b(t)$ the video bitrate being streamed at time $t$, and by $\gamma(t)$ the quotient $\frac{b_L}{b(t)}$ where $b_L$ is the highest video bitrate. We assume that, at the client side, the video frames are played at a rate of $\lambda$ frames per second (fps), and that, before starting the video, a prefetching stage is introduced till having $Q_0$ frames in the playback buffer. To avoid buffer overflows, we assume that the playback buffer is very large.

In our problem modelling, we exploit the knowledge of the user's future available throughput (hereinafter called network capacity) to optimize the system usage cost and the QoE.
Let $c(t)$ be the network future capacity at time $t$ and $r(t)$ be the transmission bitrate of the user at that time, note that $0\leq r(t)\leq c(t)$. Inspired by \cite{conf/infocom/LuV13}, we define the system utilization cost as
 \beq\label{eq:utilization}
 \sigma=\frac{1}{T}\int_0^T \frac{r(t)}{c(t)}dt,
 \eeq
where $\frac {r(t)}{c(t)}$ is the proportion of resources allocated to the user at time $t$ (can be interpreted as the proportion of time the user is occupying the network if we use discretize the time ), and $T$ defines the video length in second.
We compute the number of frames that will be streamed with quality level $j$ during the streaming session as
\beq\label{eq:qoe}
 \int _0^T \frac{\delta_{\{b(t)=b_j\}} r(t)\lambda }{b(t)} dt =\int _0^T \frac{\gamma_j(t) r(t) \lambda }{b_L} dt,
 \eeq
where
\beq
\begin{array}{cccl}
 \gamma_{j}(t) =
 \left\{
 \begin{array}{cccl}
 \displaystyle \gamma(t) & \quad\text{if} \,\,\, b(t) = b_{j}; & j \in [1\ldots L], \\
 \displaystyle 0 & \text{otherwise}. \\
 \end{array}
 \right.
\end{array}
\eeq
Assume that the user's perception on the video quality levels can be expressed by the mean of weights $\{w_1, \ldots, w_L\}$ such that $w_j$ corresponds to quality level $j$ and $w_{i} < w_{j}$ for $i <j$. Hence, we define the weighted average quality of the video as

 \beq
 \rho=\frac{\sum_{j=1}^{j=L}w_{j}\int _0^T \gamma_{j}(t)r(t) \lambda dt}{b_L \times (N \times S )} =\frac{\sum_{j=1}^{j=L}w_{j}\int _0^T \gamma_{j}(t)r(t)dt}{S_{L}},
 \eeq
where $S_{L}$ represents the video total size in bits when it is coded with the highest bitrate level $b_{L}$, i.e., $ S_L = \frac{b_L \times N \times S }{\lambda} $.

Normally, a high video quality comes at a high cost. However, it may happen that a user wishes to reduce his cost in return of a low quality, or that an operator wishes to save the network resources for further usage. To cover such situations, we define a positive parameter $\pi$ to make the tradeoff between system utilization cost and video quality. Therefore, we define our optimization cost function as
\vspace{-0.2cm}
 $$\mathcal{F}=\sigma - \pi \times \rho.$$

Let $u(t)$ be the \emph{cumulative} number of arrival frames at time $t$ and $l(t)$ be the \emph{cumulative} number of frames being already played at that time. Therefore, we define the buffer underflow constraint as $u(t) \ge l(t) \ \forall t\le T$. Given the transmission bitrate $r(t)$ and the corresponding video bitrate $b(t)$, we express the network frame rate as $\frac{\lambda \ r(t)}{b(t)}$.

Denote by $(r, \gamma)$ the video transmission strategy during the streaming session , where $r$ defines the transmission schedule and $\gamma$ characterizes the distribution of video bitrates. We start with the case where no rebuffering events will happen during the streaming session. Hence, we summarize our optimization problem, as follows

\begin{equation}\label{optim}
 \underset{(r,\gamma)}{\min} \ \mathcal{F}(r,\gamma)=
 \ \ \frac{1}{T}\int _0^T \frac{r(t)}{c(t)}dt - \pi \times \frac{\sum_{j=1}^{j=L}w_{j}\int_0^T \gamma_{j}(t)r(t) dt}{S_{L}}
 \end{equation}
\[s.t \left\{
\begin{array}{l l}
 \int_0^t \frac{ \lambda \ c(t) \gamma_{1}}{b_{L}} \ge l(t), & \forall t \le T \\
 \\
 \int_0^t \sum_{j=1}^{j=L} \frac{ \lambda \ r(t) \gamma_{j}(t)}{b_{L}} \ge l(t), & \forall t \le T\\
 \\
 \int_0^T\sum_{j=1}^{j=L} \frac{\lambda\ r(t) \gamma_{j}(t) }{b_{L}} =l(T), \\
\end{array} \right. \]
where the first constraint ensures the existence of at least one solution which corresponds to a mono-quality streaming using the lowest video bitrate and the whole resources.
At the end of Section \ref{results}, we study the case where several rebuffering events are tolerated during the streaming session.

\section{Properties of optimal solution without rebuffering events }\label{sec:Resoluion}

\subsection{The threshold scheme for transmission schedule }\label{sec:ThresholdPolicy}
\begin{definition}\label{threshold}
Giving the network capacity $c$, we define the threshold transmission schedule by
\beq
\begin{array}{cccl}
 r_{th}(t) =
 \left\{
 \begin{array}{cccl}
 \displaystyle c(t)& \quad \text{if}\ c(t) \ge \alpha \\
 \displaystyle 0 & \quad \text{otherwise.} \\
 \end{array}
 \right.
\end{array}
\eeq
\end{definition}
\begin{prop}
Assume that there exists a feasible solution that satisfies the constraints in (\ref{optim}), then there exists an optimal strategy $(r_{th}, \gamma_{r_{th}})$ of optimization problem (\ref{optim}), where $r_{th}$ is a threshold transmission schedule.
\end{prop}
This propriety was actually inspired by \cite{conf/infocom/LuV13}. Nevertheless, authors in \cite{conf/infocom/LuV13} assumed a classical video streaming with only one bitrate level, whereas we consider adaptive video streaming with multiple bitrate levels, which makes our optimization problem more appealing as it fits current video streaming schemes.

\begin{proof}
Let $c$ and $r$ be the network capacity and the user transmission bitrate on a given interval of time $[0,\epsilon]$.
 \begin{figure}[t]
 \begin{center}
 \includegraphics[width=6cm,height=4cm] {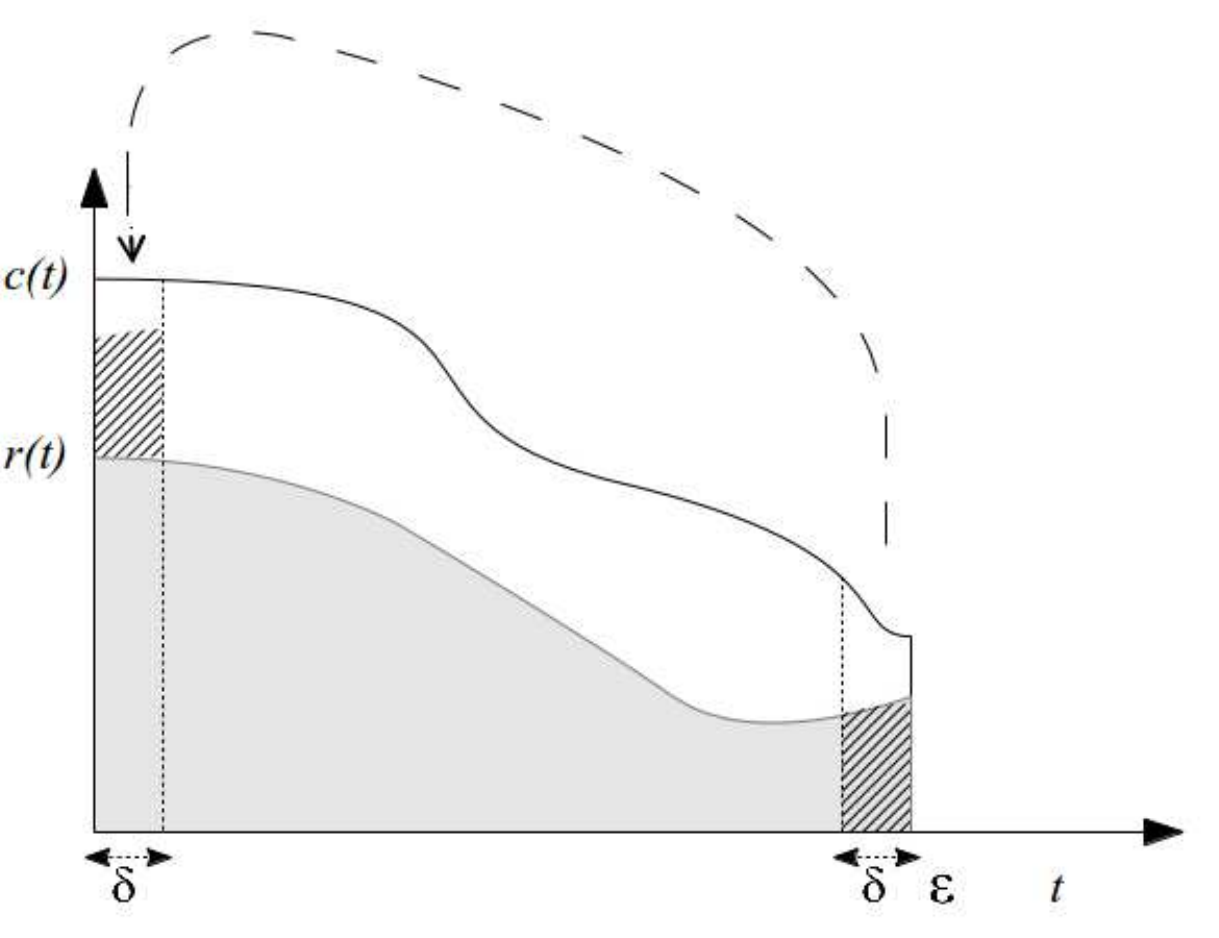}
 \vspace{-.1cm}
 \caption{Sketch of proof of the threshold strategy. Here, the hatched area on the right can be entirely shifted to the left, which gives a value of $\beta$ equal to $1$.}
 \label{fig:Demo}
 \end{center}
 \vspace*{-.5cm}
 \end{figure}
Without loss of generality \footnote{The proof still holds for a monotonically increasing $c$.} and for the sake of illustration, we choose an interval of time where $c$ is monotonically decreasing as shown in Fig. \ref{fig:Demo}. As we have $r(t) \leq c(t)\ \forall \ t \in [0,\epsilon]$, then $\exists \ (\delta,\beta) \in [0,\frac{\epsilon}{2}]\times[0,1]$ such that $\forall \ t \in [0,\delta]$
\beq\label{eq:cap-ineq}
c(t) \ge c(t+\epsilon -\delta)
\eeq
$\text{and} $
\beq \label{eq:cap-ineq-1} \int _0 ^\delta \frac{ r(t)+ \beta r(t+\epsilon - \delta)}{c(t)}dt \le \delta,
\eeq
where Inequality (\ref{eq:cap-ineq}) derives from the decreasing pace of $c$, and Inequality (\ref{eq:cap-ineq-1}) derives from the fact that some data at the end can be transmitted beforehand. On the other hand, we have
\begin{eqnarray}
\int _0^\epsilon \frac{r(t)}{c(t)}dt =
 \int _0^\delta \frac{r(t)+\beta r(t+\epsilon-\delta) }{c(t)}dt +
 \int _\delta^{\epsilon-\delta} \frac{r(t)}{c(t)}dt\nonumber \\
 \ + \int _{\epsilon-\delta}^\epsilon \frac{r(t)}{c(t)}dt
 - \int _{0} ^\delta \frac{\beta r(t+\epsilon-\delta)}{c(t)}dt.
\end{eqnarray}
Using Inequality (\ref{eq:cap-ineq}), we obtain
\begin{eqnarray}\label{eq:util-ineq}
\int _0^\epsilon \frac{r(t)}{c(t)}dt \ge
 \int _0^\delta \frac{r(t)+\beta r(t+\epsilon-\delta) }{c(t)}dt +
 \int _\delta^{\epsilon-\delta} \frac{r(t)}{c(t)}dt \nonumber\\
 + \int _{\epsilon-\delta}^\epsilon \frac{r(t)}{c(t)}dt)
 - \int _{\epsilon-\delta} ^\epsilon \frac{\beta r(t)}{c(t)}dt.
\end{eqnarray}
Obviously, if
$$ \int _0 ^\delta \frac{\ r(t)+\beta r(t+\epsilon - \delta)}{c(t)}dt = \delta,$$
then all the given capacities in $[0,\delta]$ will be used, i.e., all the white surface in Fig.~\ref{fig:Demo} will be filled. In that case, we define a new transmission schedule $r^{\prime}$ such that
\beq
\begin{array}{cccl}
 r^{\prime}(t) =
 \left\{
 \begin{array}{cccl}
 \displaystyle c(t) & \quad\mbox t \in [0,\delta] \\
 \displaystyle r(t) & \quad\mbox t \in ]\delta,\epsilon-\delta[ \\
 \displaystyle (1-\beta) r(t) & \quad\mbox t \in [\epsilon-\delta,\epsilon] \\
 \end{array}
 \right.
\end{array}
\eeq
which gives
$$ \int _0^\epsilon \frac{r(t)}{c(t)}dt \ge \int _0^\epsilon \frac{r^{\prime}(t)}{c(t)}dt $$
Otherwise, if
\beq
\int _0 ^\delta \frac{\ r(t)+\beta r(t+\epsilon- \delta)}{c(t)}dt < \delta, \eeq
then $\beta$ will be equal to $1$ since our objective is to shift as much data as possible from the times where the capacity is low to the times where the capacity is high. Therefore, to completely use the highest capacities, we must repeat the same shifting operation on $[0,\epsilon-\delta]$ considering a new transmission function $r^{\prime}$ verifying
\beq
\begin{array}{cccl}
 \left\{
 \begin{array}{cccl}
 \displaystyle \int _0 ^\delta \frac{r^{\prime}(t)}{c(t)}dt = \int _0 ^\delta \frac{\ r(t)+\beta r(t+\epsilon - \delta)}{c(t)}dt \\
 \displaystyle r^{\prime}(t)= r(t) \ \forall \ t \in [\delta,\epsilon -\delta].\\
 \end{array}
 \right.
\end{array}
\eeq
In both cases, Inequality (\ref{eq:util-ineq}) holds, which means that the highest capacities are less expensive than the lowest capacities in terms of network utilization cost if they were used for transmitting data.
If we keep repeating the shifting operation on all the future horizon, we end up having all the highest capacities entirely used and all the lowest one unused, which is clearly a threshold transmission schedule as defined in Definition \ref{threshold}.

Now, we assume that, knowing $c$, there exists a feasible solution $(r,\gamma)$ that satisfies the constraints in (\ref{optim}). To perform the data shifting operation on the transmission schedule, three main conditions should be verified:
(i) The shifted data must have the same video bitrate as the bitrate used in the shifted-to time,
(ii) data shifting shall not interrupt a segment transmission schedule,
(iii) data shifting shall not violate the stall constraints.

Actually, shifting the data transmission can be either done to the left (earlier) or to the right (later).
As we assume a very large playback buffer, sending the video data at earlier times will not cause packets rejection and, thus, will not cause video stalls. In other words, any data shifting to earlier times of higher capacities will be performed without violating the stall constraints.
However, when the higher capacity values come later, the data shifting must be checked whether it violates the stall constraints or not.
As we only shift the data transmission without changing their corresponding video bitrates, we end up having a new bitrate level strategy $\gamma_{r_{th}}$ that gives the same weighted average quality as given by $\gamma$.
Thereby, the resulting strategy $(r_{th},\gamma_{r_{th}})$ outperforms strategy $(r,\gamma)$, which completes the proof.
\end{proof}
In practice, the setting of the transmission threshold $\alpha$ does not follow the data shifting process of the proof.
We will thus design an approach to build a threshold strategy for the transmission schedule.

\subsection{Ascending bitrate level strategy}\label{sec:ascendingapproch}
In this section, we study the proprieties of the bitrate level strategy under a threshold based transmission schedule. More specifically, we analyze the impact of the video quality levels' order on the setting of $\alpha$.
\begin{definition}
We say a bitrate level strategy is {\bf ascending} if the quality levels of the video segments increases during the session, i.e., for all $0\leq t\leq t^{\prime}\leq T$
$$
b(t) \leq b(t^{\prime}), \mbox{ i.e., } \gamma(t)\geq \gamma(t^{\prime}).
$$
\end{definition}

\begin{prop}
Assume that there exists a threshold-based solution $(r_{th},\gamma)$ that satisfies the constraints in (\ref{optim}), then there exists a threshold-based ascending bitrate level solution $(r^{\prime}_{th},\gamma^{\prime})$ that optimizes problem in (\ref{optim}).
\end{prop}

\begin{proof}
Pick a suite of $N$ segments with a non ascending order quality levels, in a way that they can be streamed without video stalls over the future horizon. Then, according to this quality levels' order, set a threshold-based solution $(r_{th},\gamma)$ with threshold $\alpha$ such that, beyond this threshold, the first constraint violation will occur at time $t=s_n$.
Suppose that, under this solution, two bitrate levels $b_1$ and $b_2$ will be respectively streamed over $[\tau, \tau + \delta]$ and $[\tau^{\prime}, \tau^{\prime} + \delta^{\prime}]$ as depicted in Fig.~\ref{fig:DemoAscendingRate}, such that
\beqq
 \displaystyle \tau+\delta < s_n, \ \ \displaystyle \tau^{\prime} > s_n, \ \ \displaystyle b_1>b_2,
 \eeqq
 and
 \beqq
 \displaystyle \int _{\tau}^{\tau+\delta} {r_{th}(t)}dt = \int _{\tau^{\prime}}^{\tau^{\prime}+\delta^{\prime}} {r_{th}(t)}dt.
\eeqq
 \begin{figure}[t]
 \begin{center}
	\includegraphics [width=6.5cm,height=3.5cm] {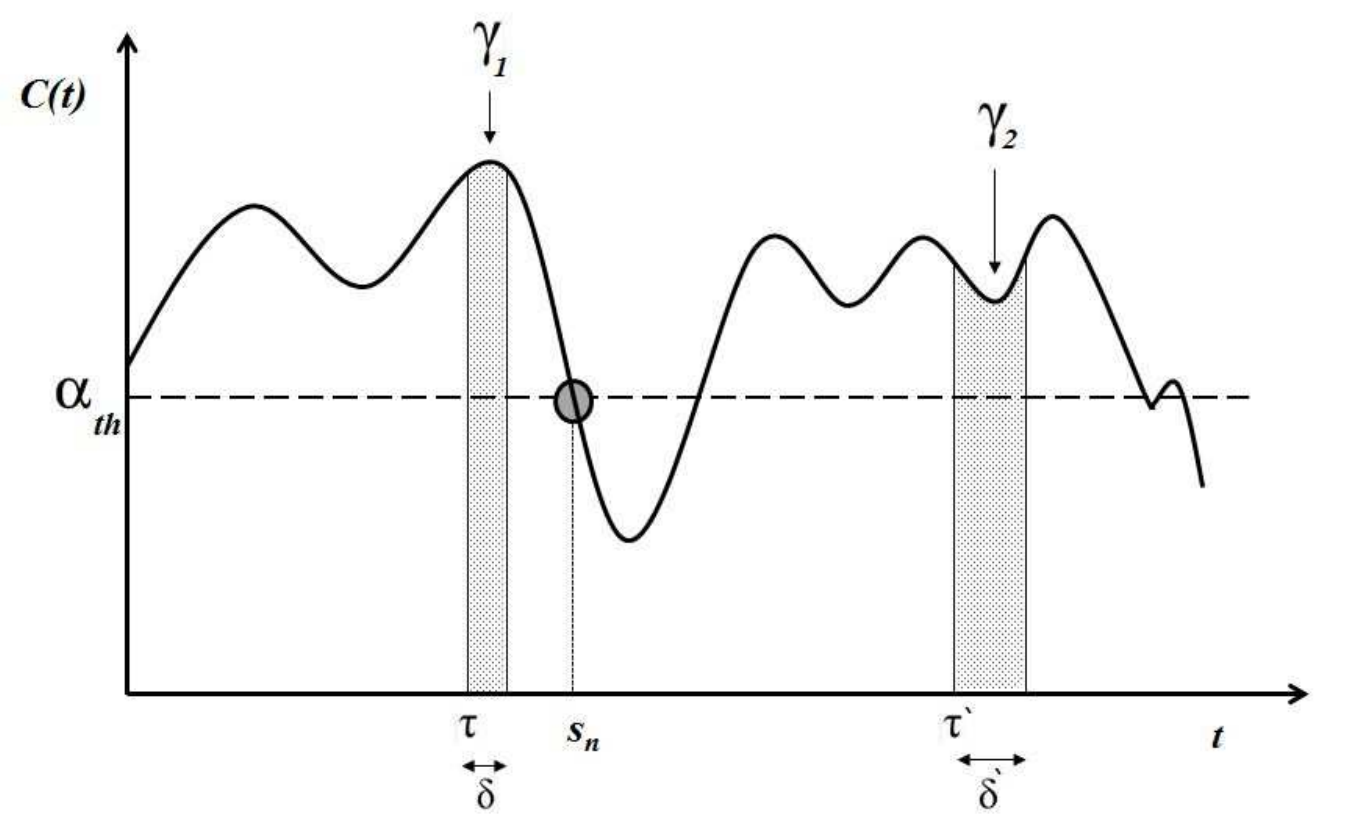}
 \vspace{-.1cm}
	\caption{Ascending bitrate level strategy.}
	\label{fig:DemoAscendingRate}
 \end{center}
 \vspace*{-.5cm}
 \end{figure}
Let $fr_{th}(t)$ be the network frame rate at time $t$. As we have $b_1>b_2$, then the number of frames that will be streamed during $[\tau^{\prime}, \tau^{\prime} + \delta^{\prime}]$ is greater than the number of frames that will be streamed during $[\tau, \tau + \delta]$. Therefore, $\exists \ \beta > 0 $ such that
 \beq
 \int _{\tau^{\prime}}^{\tau^{\prime}+\delta^{\prime}} {fr_{th}(t)}dt = \int _{\tau}^{\tau+\delta} {fr_{th}(t)}dt + \beta.
 \eeq
Suppose that we switch between $b_1$ and $b_2$ over these two intervals of time. Then, the number of cumulative received frames at $s_n$ will be increased by $\beta$. Let $u$ and $u^\prime$ be the cumulative number of arrival frames functions before and after switching the bitrates. therefore, we have
\beq
u^{\prime}(s_n) = u(s_n) + \beta.
\eeq

Actually, if $u^{\prime}(s_n)$ is large enough and allows increasing the threshold beyond $\alpha$ without violating the stall constraint at $t=s_n$ and later, then the cost function will be reduced. Otherwise, the threshold remains the same without changing the system performance. In fact, as explained in the previous section, streaming the data beforehand will only add more flexibility toward the stall constraints since the buffer is assumed to be very large. We show by the sequel that, even if we switch between the two bitrate levels the streaming will remain without video stalls under the same threshold since $u^{\prime}\ge u(t)\ \forall t \in [0,T]$ (see Fig.~\ref{fig:switch}).
Let $fr^{\prime}_{th}$ be the network frame rate function after switching. Then, we have
\beq
fr^{\prime}_{th}(t)\ > fr_{th}(t) \ \forall t \in [\tau,\tau+\delta[,
\eeq
\beq \label{eq17}
fr^{\prime}_{th}(t)\ < fr_{th}(t) \ \forall t \in [\tau^{\prime},\tau^{\prime}+\delta^{\prime}[,
\eeq
\beq \label{eq18}
\int_{\tau} ^ {\tau+\delta} fr^{\prime}_{th}(t)-fr_{th}(t)\ dt = \int_{\tau^{\prime}} ^ {\tau^{\prime}+\delta^{\prime}} fr_{th}(t)-fr^{\prime}_{th}(t)\ dt = \beta.
\eeq
We further define $u^{\prime}$ as
\beq
\begin{array}{cccl}
 u^{\prime}(t) =
 \left\{
 \begin{array}{ll}
 \displaystyle u(t) & \quad\mbox t < \tau \\
 \displaystyle u(\tau) + \int_\tau ^ t fr^{\prime}_{th}(s)\ ds & \quad\mbox t \in [\tau,\tau+\delta[ \\
 \displaystyle u(t) + \beta & \quad\mbox t \in [\tau+\delta, \tau^{\prime}[ \\
 \displaystyle u(\tau^{\prime}) + \beta + \int_{\tau^{\prime}} ^ t fr^{\prime}_{th}(s)\ ds & \quad\mbox t \in [\tau^{\prime},\tau^{\prime}+\delta^{\prime}[\\
 \displaystyle u(t) & \quad\mbox t \ge \tau+\delta^{\prime}.\\
 \end{array}
 \right.
\end{array}
\eeq
 \begin{figure}[t]
 \begin{center}
	\includegraphics [width=6.5cm,height=3.5cm] {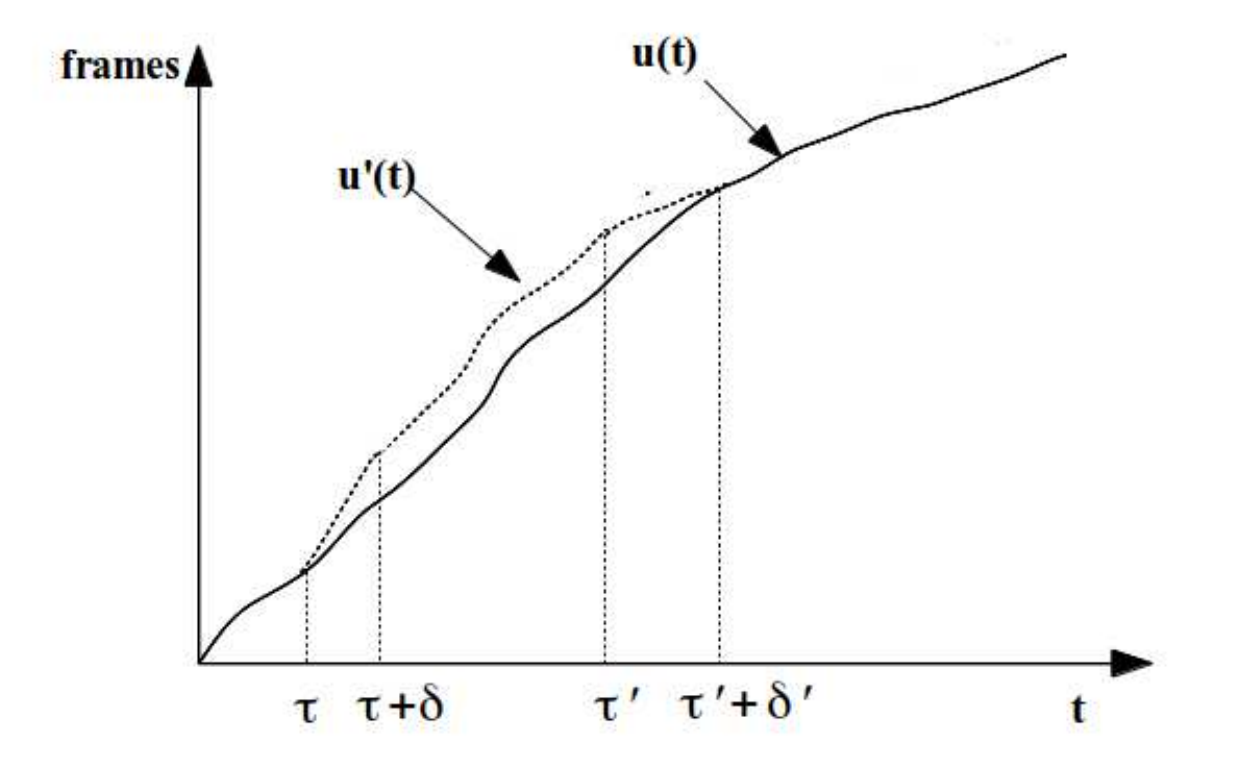}
 \vspace{-.1cm}
	\caption{Impact of bitrates switching on the cumulative number of arrival frames $u$.}
	\label{fig:switch}
 \end{center}
 \vspace*{-.5cm}
 \end{figure}
Actually, the cumulative watched frames function $l$ will remain the same as the playback frame rate $\lambda$ remains the same for all bitrate levels. Now, we see clearly that $\forall$ $t \not \in [\tau^{\prime},\tau^{\prime}+\delta^{\prime}[ \ , \ u^{\prime}(t)\ge u(t)$. However, for $t \in [\tau^{\prime},\tau^{\prime}+\delta^{\prime}[$, we have
\beq
u^{\prime}(t)-u(t) = \beta - \int_{\tau^{\prime}} ^ t fr_{th}(s)-fr^{\prime}_{th}(s)\ ds,
\eeq
which is positive according to (\ref{eq17}) and (\ref{eq18}).
To conclude, putting the segments in an ascending bitrates' order may allow a higher transmission threshold which further reduces the cost function without degrading the average quality of the video.
\end{proof}

\section{Algorithmic approaches under no rebuffering events
constraint } \label{algo}
In this section we solve optimization problem (\ref{optim}) through algorithmic approaches based on the properties of the optimal solution characterised in the previous section. We provide an approach that compute the optimal threshold-based solution but this algorithm is faced with a high computational complexity necessary to obtain the optimal solution.  Due to this shortcoming, we propose an alternative heuristic approaches to obtain nearly optimal solutions under the assumption of no rebuffering events during the session. Afterwards, we extend the study to the case where the number of video playback stalls can be tolerated to a certain level.

\subsection{ Optimal threshold-based solution }
We summarize here our global optimal approach in three main steps as illustrated in Fig.~\ref{fig:GenericAlgo}: (i) first, we look for all the possible values of $\alpha \in [\alpha_{min},\alpha_{max}]$ that satisfy the constraints in (\ref{optim}) and associate to each one the birates level strategy that gives the highest possible weighted average quality, (ii) for each threshold and its corresponding video quality, we compute the resulting cost function $\mathcal{F}$, (iii) the optimal solution corresponds to the one that minimizes $\mathcal{F}$. 
 \begin{figure}[h]
 \begin{center}
 	\includegraphics [width=8.5cm] {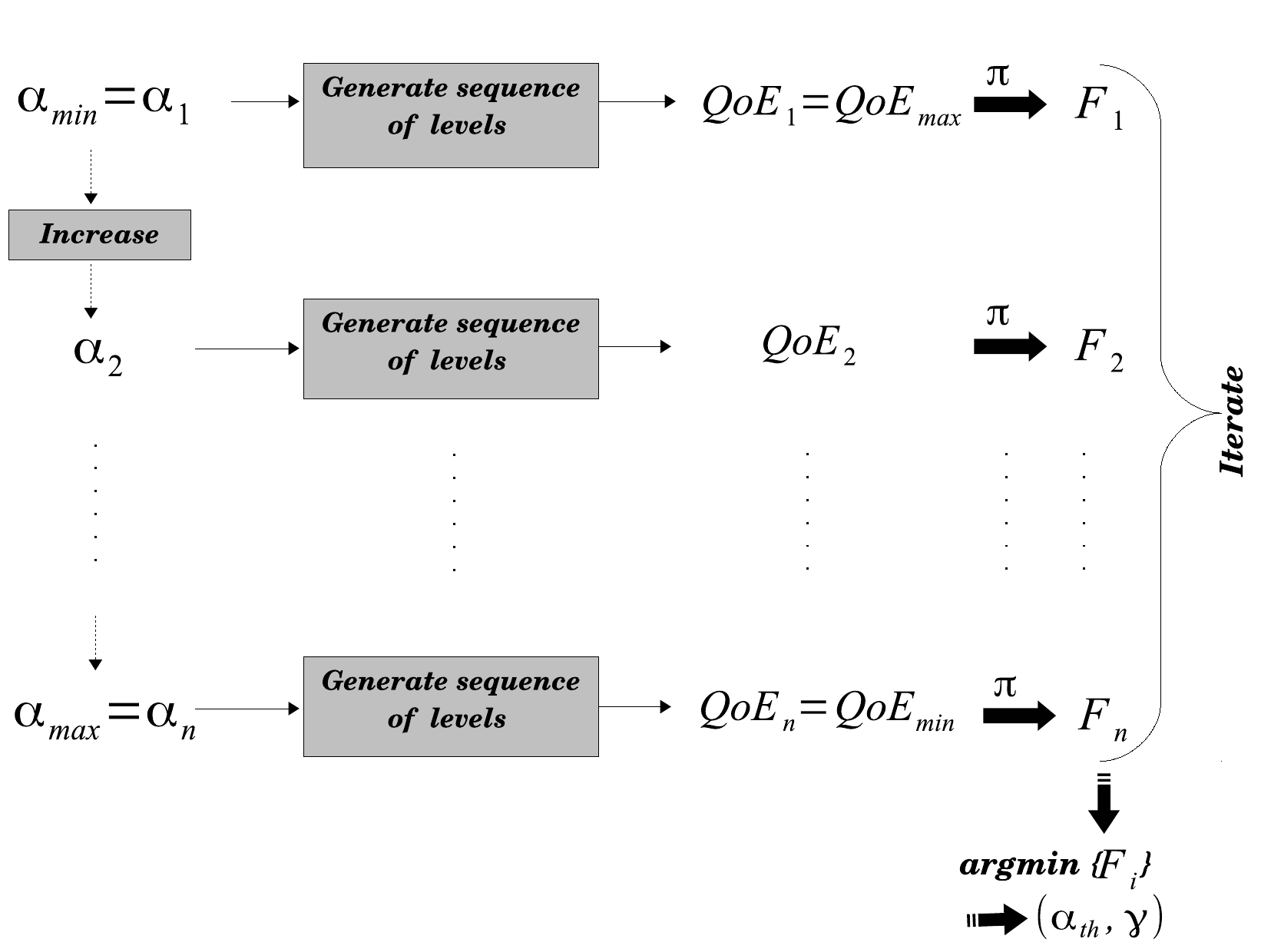}
  \vspace{-.1cm}
	\caption{Global algorithm for an optimal threshold-based solution with ascending bitrates.}
	\label{fig:GenericAlgo}
\end{center}
\vspace{-.4cm}
 \end{figure}
\subsubsection{Optimal transmission schedule $\alpha$}
To find the optimal threshold $\alpha$ with the lowest complexity, we propose to sort the future capacity in an ascending way, then try its ascendent values as thresholds till reaching the one that causes video stalls. This approach will determine all the possible thresholds $[\alpha_{min},\alpha_{max}]$. Fig.~\ref{fig:c} illustrates the example used for the simulation section. \\

\subsubsection{Optimal bitrate level strategy} Our approach for generating an optimal ascending bitrate level strategy consists of using a tree of choice of $N$ levels as depicted in Fig.~\ref{fig:tree}, where each level corresponds to a video segment.
The nodes of a tree level $i$ correspond to all possible quality levels that can be assigned to segment $i$. The parent of a node (if it exists) has either a worse or equal quality. The children (if they exist) have either a better or equal quality. We construct the tree level by level to form the path that gives the optimal sequence of bitrates. At each level, we remove the nodes whose paths cause a constraint violation in order to minimize the number of nodes at the bottom of the tree. At each level, we compute the partial weighted average quality till reaching the end of the tree. The optimal sequence of bitrates corresponds to the path that maximizes the total weighted average quality. The complexity of this algorithm may reach up to $\mathcal{O}((L+1)^N)$, which makes it non suited for online streaming services.
 \begin{figure}
 \begin{center}
 	\includegraphics [width=8cm] {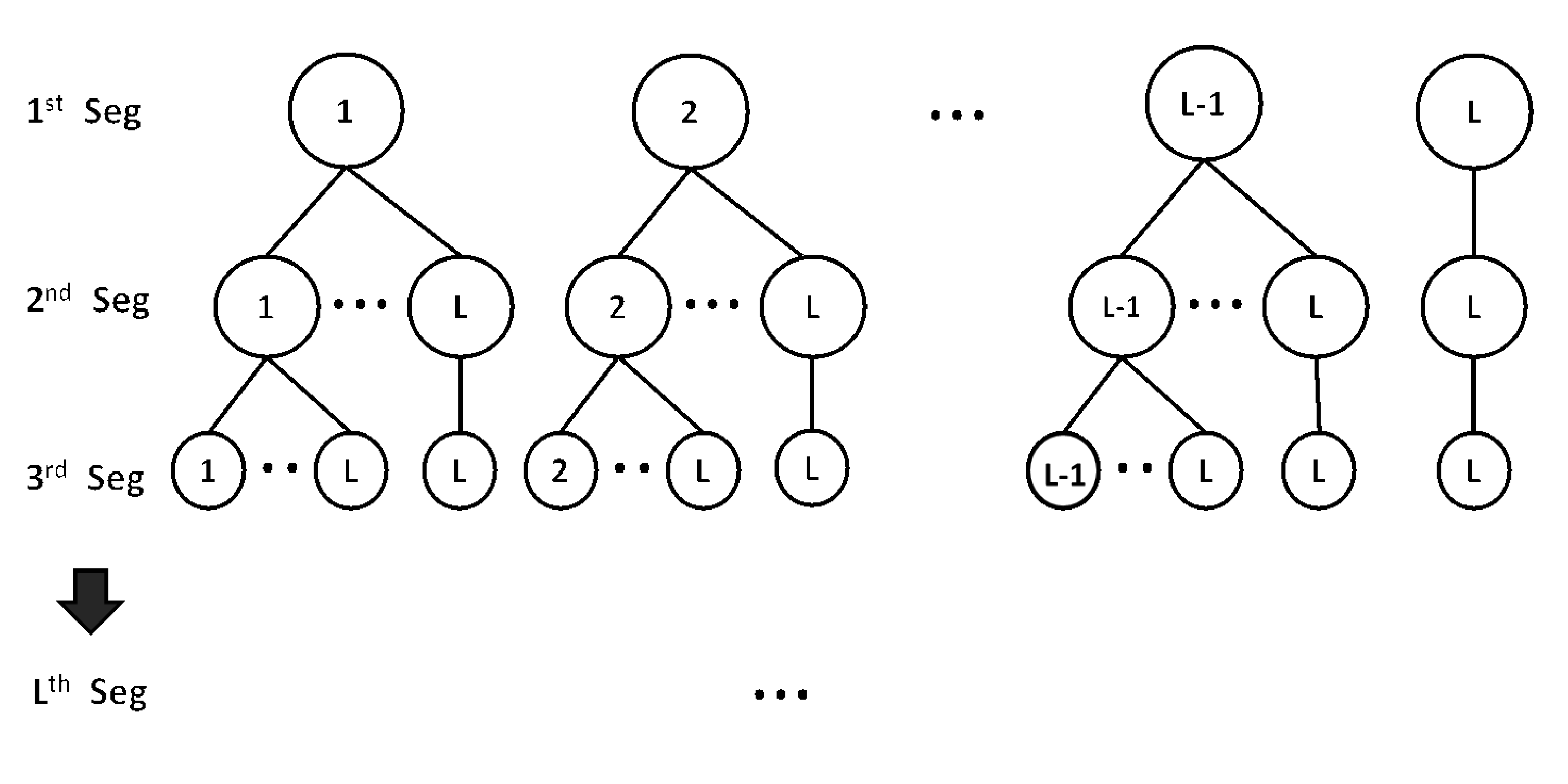}
  \vspace{-.1cm}
	\caption{Tree of choice for optimal ascending bitrate.}
	\label{fig:tree}
 \end{center}
 \vspace{-.4cm}
 \end{figure}
 \subsection{ NEWCAST design } \label{heuristic}
NEWCAST (aNticipating qoE With threshold sCheme And aScending biTrate levels) follows the same principle as the optimal global approach, but it uses two heuristics INVEST and AWARE for respectively computing the thresholds and generating the sequence of bitrates.
Let $ \gamma_\alpha$ and $\mathcal{F}_\alpha$ be the ascending bitrate levels strategy and the cost function under $r_\alpha$-based transmission schedule.
The main steps of this heuristic are described in Algorithm \ref{GeneralAlgo}. \\
 \subsubsection{INVEST: INcrease with VariablE foot STep }
This heuristic also follows the same principle as the optimal approach. However, instead of trying all the sorted capacity values as thresholds till violating the contraints, it defines a variable foot step to increase the threshold initially set to $c_{min}$. The values taken by this foot step will depend on the dynamic of the network capacity; Let $\{ \alpha_1, \ldots, \alpha_M\} \subset [\alpha_{min}, \alpha_{max}]$ such that $\alpha_{i+1} > \alpha_{i}, i\in \{1, \ldots, M-1\}$. To compute $\alpha_{i+1}$ knowing $\alpha_{i}$, we set the number of bits that we want to abandon through increasing the threshold (denoted by $Q$). Then, we find the capacity value (threshold) that allows doing that as described in Fig.~\ref{fig:demoQ}. $\alpha_{i+1}-\alpha_i$ will define the $i^{th}$ foot step (See Algorithm \ref{Footstep}).

\begin{algorithm}[t]
\caption{INVEST: INcrease with VariablE foot STep}
 \KwData {$c$, $i$, $Q$}
\SetAlgoLined\DontPrintSemicolon
\nl SortedC=sort($c$);  \;
\nl CumSortedC=CumulativeSum(SortedC); \;
\nl ind = $\max$(find (CumSortedC $\le i \times Q$)); \;
\nl \KwRet\ SortedC(ind)
 \vspace{.3cm}
\label{Footstep}
\end{algorithm}

 \begin{figure}[h]
 \begin{center}
 	\includegraphics [height=4cm,width=7cm] {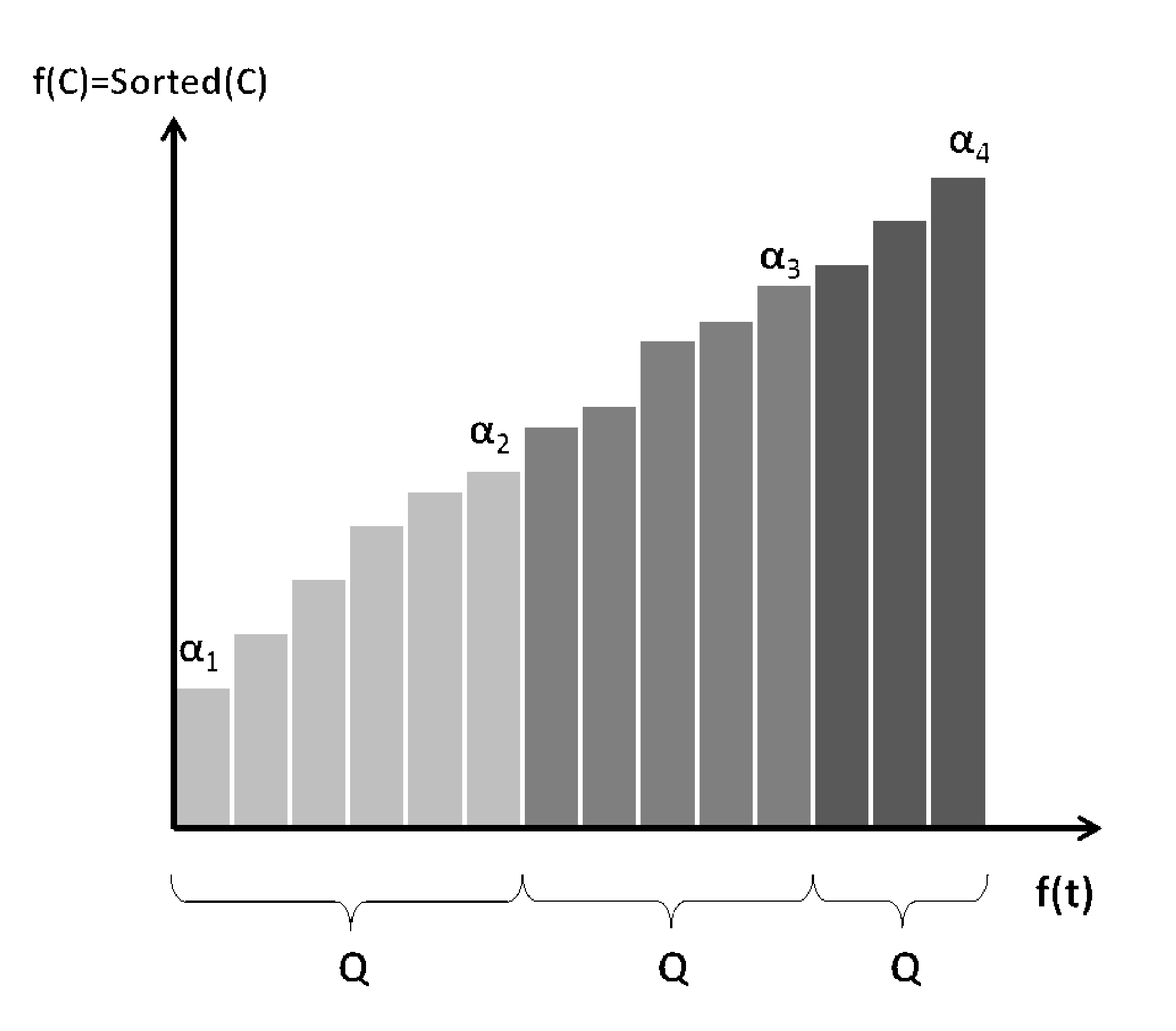}
  \vspace{0cm}
	\caption{INVEST: INcrease with VariablE foot STep.}
	\label{fig:demoQ}
\end{center}
\vspace{-.4cm}
 \end{figure}

\subsubsection{AWARE : Anticipating qoe With Ascending bitRate lEvels} This heuristic has a polynomial complexity and is quite faster than the optimal approach. Our simulation results show that its outcoming solution approaches the optimal solution at almost $98\%$ in terms of the video average quality. We summarize its steps in the few following points: \\
At the beginning, we assign the lowest bitrate to all video segments. Then, starting from the end of the video (latest segment) back to the beginning, we increase the bitrate of each segment by one level as long as the stall constraints are satisfied. We repeat this step many times till reaching the highest available bitrate (See Fig.~\ref{fig:Levels}).
By following this approach, the number of times the bitrate will be increased is at most equal to $L-1$ (see Algorithm~\ref{Heuristic}). To reduce the startup delay, which is a prominent key QoE factor (but not included in our optimization problem), we set the startup-segments to the lowest bitrate and stream them using a greedy\footnote{A greedy transmission uses all the available network capacities.} transmission rather than a threshold-based transmission. As shown in Fig.~\ref{fig:compareOptHeuristic}, an inherent advantage of this algorithm is that it ensures a progressive increase of the bitrate instead of an aggressive increase as given by the optimal approach, which is quite more appreciated by the users.

\begin{figure}[t]
 \begin{center}
  \includegraphics[width=7cm] {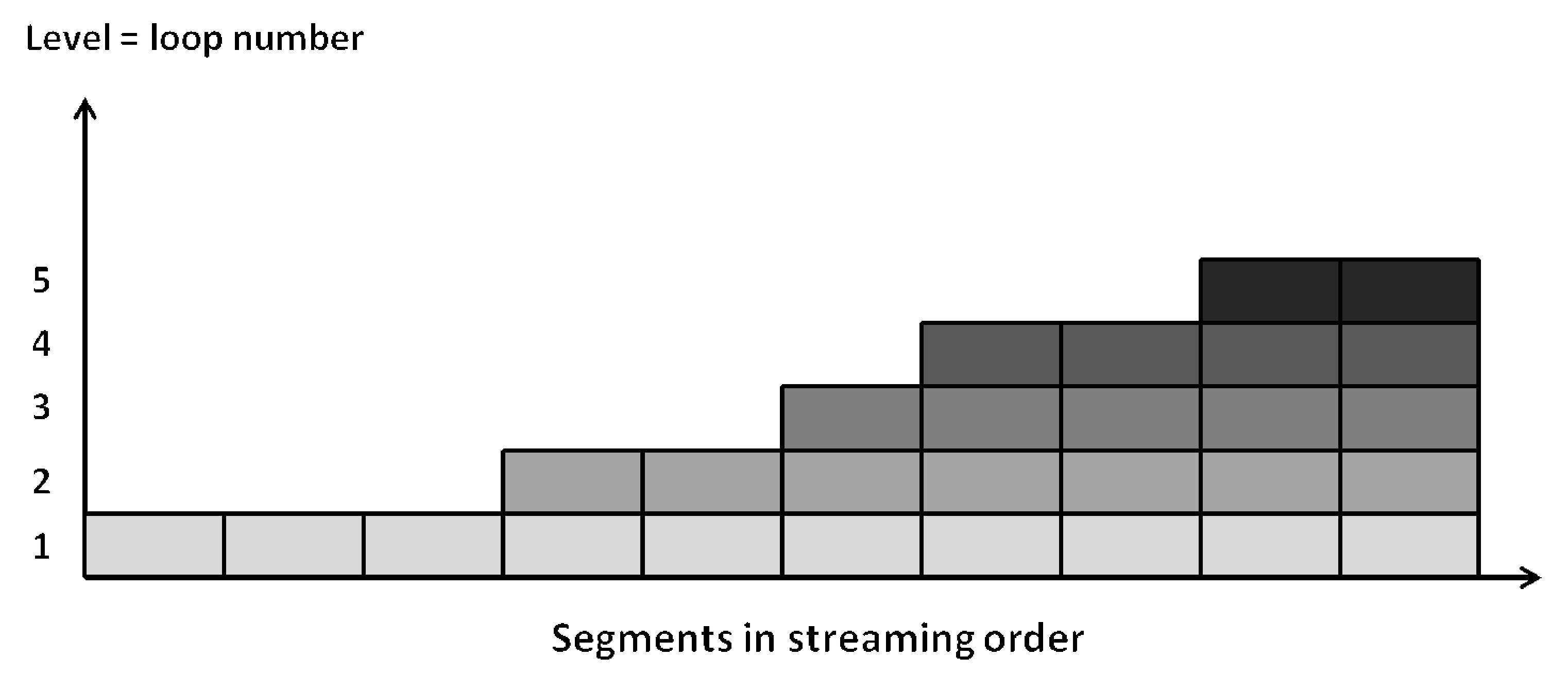}
  \caption{Illustrative example of the bitrate increasing steps used in AWARE.}
  \label{fig:Levels}
 \end{center}
 \end{figure}

 \begin{figure}[t]
 \begin{center}
 	\includegraphics[height=3cm]{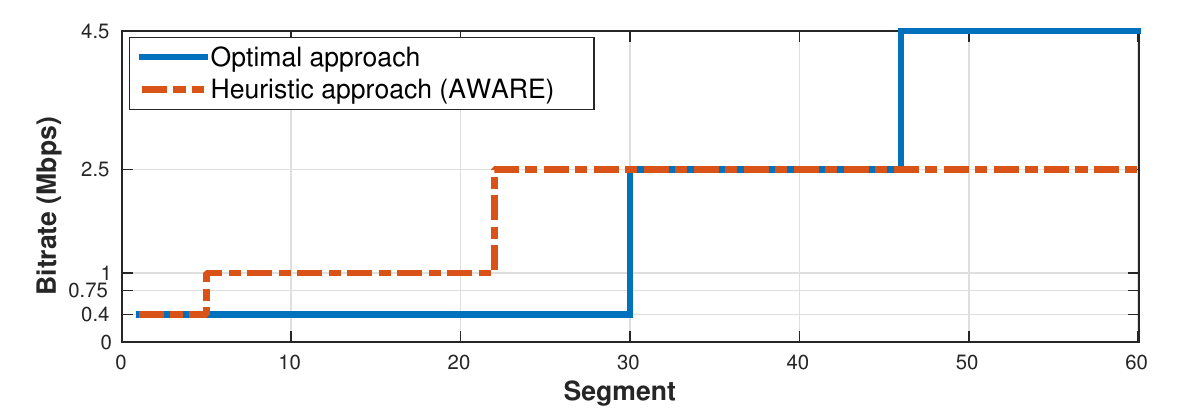}
  \vspace{-0.0cm}
	\caption{Comparative example between optimal approach and AWARE.}
	\label{fig:compareOptHeuristic}
\end{center}
\vspace{-0.5cm}
 \end{figure}

\begin{algorithm}[h]
\SetAlgoLined\DontPrintSemicolon
\KwData {$c$, VideoProperties, $L, w, Q$}
\nl $\alpha$=$c_{min}$; $i=1$; \;
\nl [PossibleTransmission, $r_\alpha$, $\gamma_{\alpha}$]=\textbf{\textbf{AWARE}}($c$, $\alpha$, videoProperties, $L$); \;
\nl \While{PossibleTransmission }{
\nl $\mathcal{F}_{\alpha}$=computeObjFunction ($c,r_{\alpha},\gamma_{\alpha},w$); \;
\nl i=i+1; \;
\nl $\alpha$ = \textbf{\textbf{INVEST}}($c, i, Q$); \;
\nl [PossibleTransmission, $r_\alpha$, $\gamma_{\alpha}$]=\textbf{\textbf{AWARE}}($c$, $\alpha$, videoProperties, $L$); \;
 }
 \nl $\mathcal{F}^{*}_{\alpha^{*}}$=min$\{ \mathcal{F}_{\alpha}\}$; \;
 \nl $\alpha_{th}$=$\alpha^{*}$; \;
\nl \KwRet\ ($\alpha_{th}$,$\gamma_{\alpha_{th}}$)
 \setcounter{AlgoLine}{0}
 \vspace{.3cm}
\caption{NEWSCAST: aNticipating qoE With threshold sCheme And aScending biTrate levels.}
\label{GeneralAlgo}
\end{algorithm}

\begin{algorithm}[h]
\SetAlgoLined\DontPrintSemicolon
\KwData {$c,\alpha$, videoProperties, $b_1 \ldots b_L$}
\nl $s=1$; SegmentsBitrates[$1:N$]=$b_s$; \;
\nl \While{$s < L$}{
\nl $s=s+1$; \;
\nl Start=FirstSegmentOfBitrate($b_{s-1}$); \;
\nl End=N; \;
\nl middle = (End-Start) div2 +1; \;
 \nl \While{middle $\ge$ 1 and End $\ge$ Start and middle $\leq$ End ) } {
 \nl init=SegmentsBitrates; \;
 \nl SegmentsBitrates[middle:End]=$b_s$ ; \;
 \nl SegmentsBitrates[1:StartupSegments]=$b_1$ ; \;
 \nl Test = ExistViolation(SegmentsBitrates,$c,\alpha$,videoProperties); \;
 \nl \eIf {Test}{
 \nl SegmentsBitrates[middle:End] = init[middle:End]; \;
 \nl middle=middle+(End-middle) div2 +1 ; \;
 }{
 \nl End=middle-1; \;
 \nl middle=Start+(End-Start) div2 +1;\;
 }
}
 }
\nl [$r_{\alpha}$,$\gamma_{\alpha}$]=TransmitVideo($c$, $\alpha$, VideoProperties, SegmentsBitrates);\;
\nl Test = ExistViolation(SegmentsBitrates, $c$, $\alpha$,VideoProperties); \;
\nl \KwRet\ ($\overline{Test}$, $r_{\alpha}$, $\gamma_{\alpha}$) {}
 \setcounter{AlgoLine}{0}
 \vspace{.5cm}
\caption{AWARE: Anticipating QoE With Ascending bitRate lEvels.}
\label{Heuristic}
\end{algorithm}

\section{Algorithmic approaches under rebuffering events}\label{algo2}
So far, we have assumed no rebuffering events during the streaming session, i.e., the future capacity has been assumed quite sufficient to allow streaming the hole video at the lowest bitrate. In extreme cases, the capacity may not be sufficient and may cause the player having video stalls even with the lowest quality level. End-user may prefer to tolerate few stalls in order to have a better quality. To go further with the analysis, we adapt our approach to a similar case where $q$ stalls can be tolerated during a session. The optimization problem in (\ref{optim}) becomes

\begin{equation}\label{optim1}
 \underset{(r,\gamma)}{{\min}} \ \mathcal{F}(r,\gamma)=
 \ \ \frac{1}{T}\int _0^T \frac{r(t)}{c(t)}dt - \pi \times \frac{\sum_{j=1}^{j=L}w_{j}\int_0^T \gamma_{j}(t)r(t) dt}{S_{L}},
 \end{equation}
\[s.t \left\{
\begin{array}{l l}
 \int_0^T\sum_{j=1}^{j=L} \frac{\lambda\ r(t) \gamma_{j}(t) }{b_{L}} =l(T), \\
{\cal{F}}_{(r,\gamma)}(T) \leq q,
 \end{array} \right. \]
where ${\cal{F}}_{(r,\gamma)}(T)$ is the number of stalls during the streaming session under strategy $(r,\gamma)$.

\begin{lemma}
\label{lem1}
Any optimal strategy will experience exactly $q$ stalls.
\end{lemma}
\begin{proof}
Assume that there exists an optimal solution $(r,\gamma)$ that has experienced $q^{\prime}$ stalls such that $q^{\prime}<q$. Suppose that under $(r,\gamma)$, $x_0$ frames have been downloaded over $[\tau,\tau+\bar{x}_0]$ where $\bar{x}_0$ is the time needed to download $x_0$ frames. By imposing an additional starvation at time $\tau$, the number of cumulative frames at playout buffer will be increased by $x_0$. This allows to give more opportunity for the transmission schedule to reduce the cost of transmission without changing their corresponding video bitrate $\gamma$ and without violating the stall constraints. Thus, the strategy $(r,\gamma)$ may decrease the cost function ${\cal{F}}$ by forcing an additional stall, which completes the proof.
 \end{proof}
The following result extends the proprieties of the optimal solution by including the possibility of rebuffering.
By Lemma \ref{lem1}, the following corollary holds.
\begin{corol}
\label{corol1}
Assume that there exists a feasible solution that satisfies the constraints in (\ref{optim1}), then there exists an optimal strategy $(r_{th} , \gamma_{th})$ of optimization problem (\ref{optim1}), where $r_{th}$ is a threshold transmission schedule and $\gamma_{th}$ is a threshold-based ascending bitrate level solution.
\end{corol}
With the above results, the algorithmic approaches under no rebuffering events still hold for the general case where the number of video playback stalls can be tolerated.

In Algorithm \ref{alg:Heuristic2}, we present the modified NEWCAST algorithm where we allow video playback stalls to happen. The major modification concerns only AWARE algorithm to compute the optimal bitrate level strategy since INVEST algorithm remains unchanged under rebuffering events.

 \begin{algorithm}[h]
\SetAlgoLined\DontPrintSemicolon
 \KwData {$c$, $\alpha$, videoProperties, $b_1 \ldots b_L$, maxStalls=$q$}
\nl $s=1$; SegmentsBitrates[1:N]=$b_s$; \;
\nl \While{s $<$ L}{
\nl $s=s+1$; \;
\nl Start=FirstSegmentOfBitrate($b_{s-1}$); \;
\nl End=N; \;
\nl middle = (End-Start) div2 +1; \;
 \nl \While{middle $\ge$ 1 and End $\ge$ Start and middle $\leq$ End ) } {
 \nl init=SegmentsBitrates; \;
 \nl SegmentsBitrates[middle:End]=$b_s$ ; \;
 \nl SegmentsBitrates[1:StartupSegments]=$b_1$ ; \;
 \nl nbrStalls = ComputeViolations(SegmentsBitrates, $c$, $\alpha$, videoProperties); \;
 \nl \eIf {nbrStalls $>$ maxStalls}{
 \nl SegmentsBitrates[middle:End] = init[middle:End]; \;
 \nl middle=middle+(End-middle) div2 +1 ; \;
 }{
 \nl End=middle-1; \;
 \nl middle=Start+(End-Start) div2 +1;\;
 }
}
 }
\nl [$r_{\alpha}$,$\gamma_{\alpha}$]=TransmitVideo($c$ , $\alpha$, VideoProperties, SegmentsBitrates);\;
\nl nbrStalls = ComputeViolations(SegmentsBitrates, $c$, $\alpha$, VideoProperties); \;
\nl Test= nbrStalls $\le$ maxStalls; \;
\nl \KwRet\ (Test, $r_{\alpha}$, $\gamma_{\alpha}$) {}
 \setcounter{AlgoLine}{0}
 \vspace{.5cm}
\caption{AWARE-MS$_q$: AWARE with at Maximum $q$ Stalls.}
\label{alg:Heuristic2}
\end{algorithm}

\section{Simulations and numerical results} \label{sec:results}
\subsection{Simulation tools and setup}
We performed all our simulations using Matlab server R2015b on a Dell PowerEdge T420 Intel Xeon running Ubuntu 14.04. The streaming session was configured according to some DASH and Youtube parameters \cite{YoutubeResolutions,Lederer:2012:DAS:2155555.2155570} and the network capacity was randomly generated around a constant average throughput.
To the best of our knowledge, no explicit way does really exist to compute the weights that can be accorded to the video bitrates. In \cite{DBLP:conf/icc/ShenLLY14}, authors were exploring a QoE estimation model in which they were assigning to each video segment a QoE metric with a logarithmic variation as function of the bitrate and the motion factor. In \cite{conf/icc/EssailiSSSKS13}, however, authors used a MOS (Mean Opinion Score) factor in order to reflect the user's satisfaction toward each quality level.
In this paper, we assign the weights to the bitrates in a proportional way as follows
$$w_i =\frac{b_i}{\sum \limits_{{i=1}}^L b_i},$$ where $b_i$ is the $i^{th}$ bitrate level and $w_i$ is its corresponding weight. All the parameters are listed in Table \ref{parameters}.
For the sake of accuracy, we explore the values of the threshold $\alpha$ using the optimal approach. Our heuristic (INVEST) will be discussed later in Section~\ref{complexity}.

\begin{table}
\begin{center}
\begin{tabular}{|c|c|}
\hline
\textbf{Window Size} & 3 min 10 s \\
\hline
\textbf{Average throughput} & 2 Mbps \\
\hline
\textbf{Capacity Time Slot} & 1 s \\
\hline
\textbf{Video Length} & 3 min \\
\hline
\textbf{Segment Length} & 1s \\
\hline
\textbf{Video frame rate} & 30 fps \\
\hline
\textbf{Playback cache} & 4s \\
\hline
\textbf{Video bitrates (Mbps)} & [0.4 0.75 1 2.5 4.5] \\
\hline
\textbf{Levels weights} & [0.09 0.17 0.22 0.55 1] \\
\hline
\end{tabular}
	\end{center}
	\caption{Parameters of Matlab simulations.}
	\label{parameters}
\vspace{-0.5cm}
\end{table}

\subsection{Framework performance}\label{results}
Fig. \ref{fig:c} outlines the dynamic of the capacity we used for all the simulation section and its correspondent values of the threshold $\alpha$. Note that, when $\alpha$ exceeds its maximum value, a stall constraint will be violated. By the sequel, we define our benchmark as the case where all the future capacity is used and the highest possible video quality is delivered, i.e., $\alpha=c_{min}$.
The execution of NEWCAST using the above parameters showed us a variation in the system performance for $\pi$ ranging from $4.50$ to $4.70$. Beyond the limits of this interval, the system performance remained constant.
In the following analysis, we will only focus on three values of $\pi$: \textit{low}, \textit{medium} and \textit{high}.
Denote by $\alpha_{\pi}$ the outcoming threshold after running NEWCAST using the preference parameter $\pi$.

In Fig.~\ref{fig:3curves} we show the variation of $\alpha_{\pi}$ as function of $\pi$; a small value of $\pi$ results in a high $\alpha_{\pi}$ as it prioritizes the system utilization cost. A big value of $\pi$, however, results in a low threshold as it accords more importance to the average quality. Whereas a medium $\pi$ leads to an in-between threshold that balances QoE and system cost.

 \begin{figure}[ht]
 \begin{center}
 	\includegraphics [height=5cm,width=9cm] {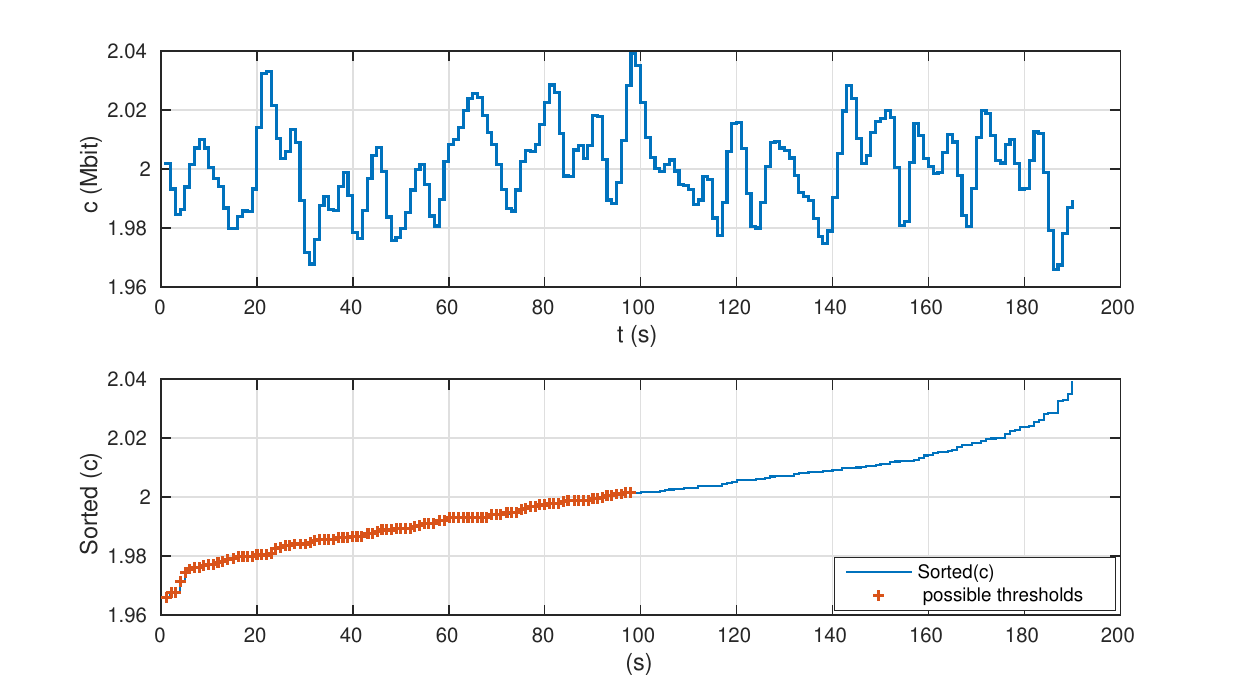}
 \vspace{-.7cm}
	\caption{Network capacity and threshold $\alpha$.}
	\label{fig:c}
\end{center}
 \end{figure}

 \begin{figure}[ht]
 \begin{center}
 	\includegraphics[width=9.5cm, height=4cm ] {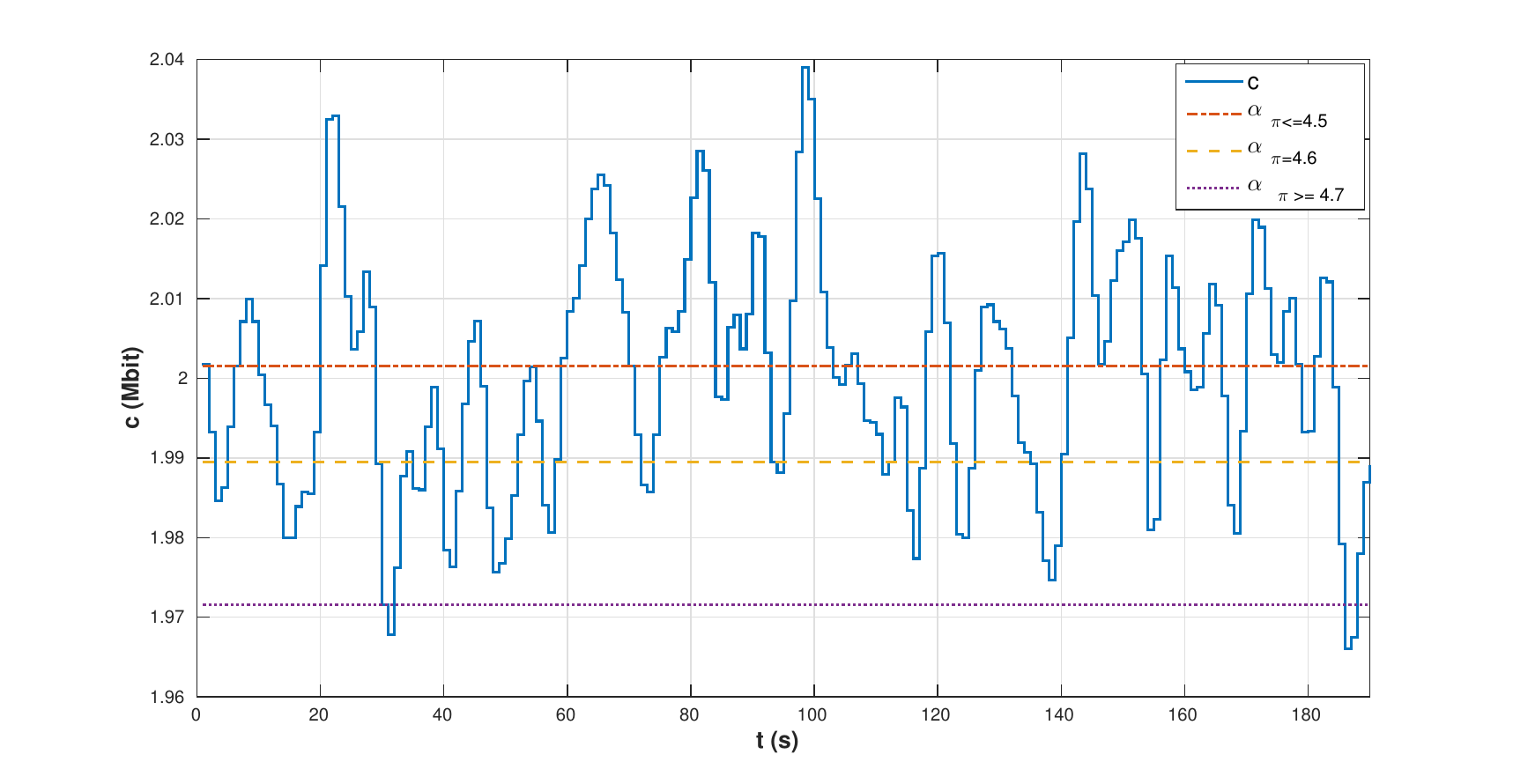}
	\vspace{-.6cm}
	\caption{Variation of $\alpha_{\pi}$ as function of $\pi$.}
	\label{fig:3curves}
\end{center}
\end{figure}

\begin{figure}[ht]
 \begin{center}
 	\includegraphics [width=9cm] {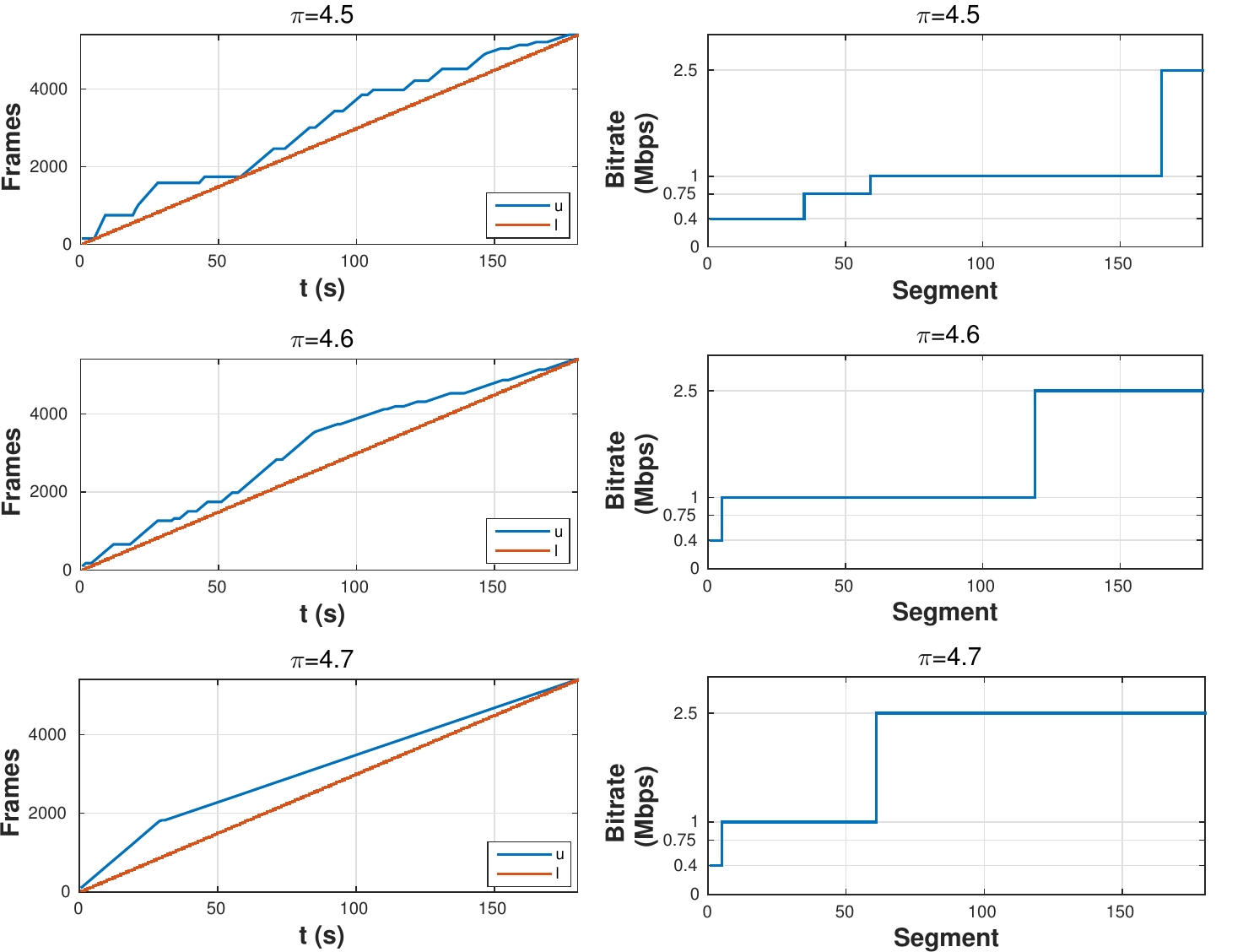}
	\vspace{-.5cm}
	\caption{Playback buffer state evolution and corresponding sequence bitrates for different $\pi$.}
	\label{fig:constraints}
\end{center}
\vspace{-.5cm}
\end{figure}

\begin{figure*}[t]
\begin{center}
 	\includegraphics [width=17cm] {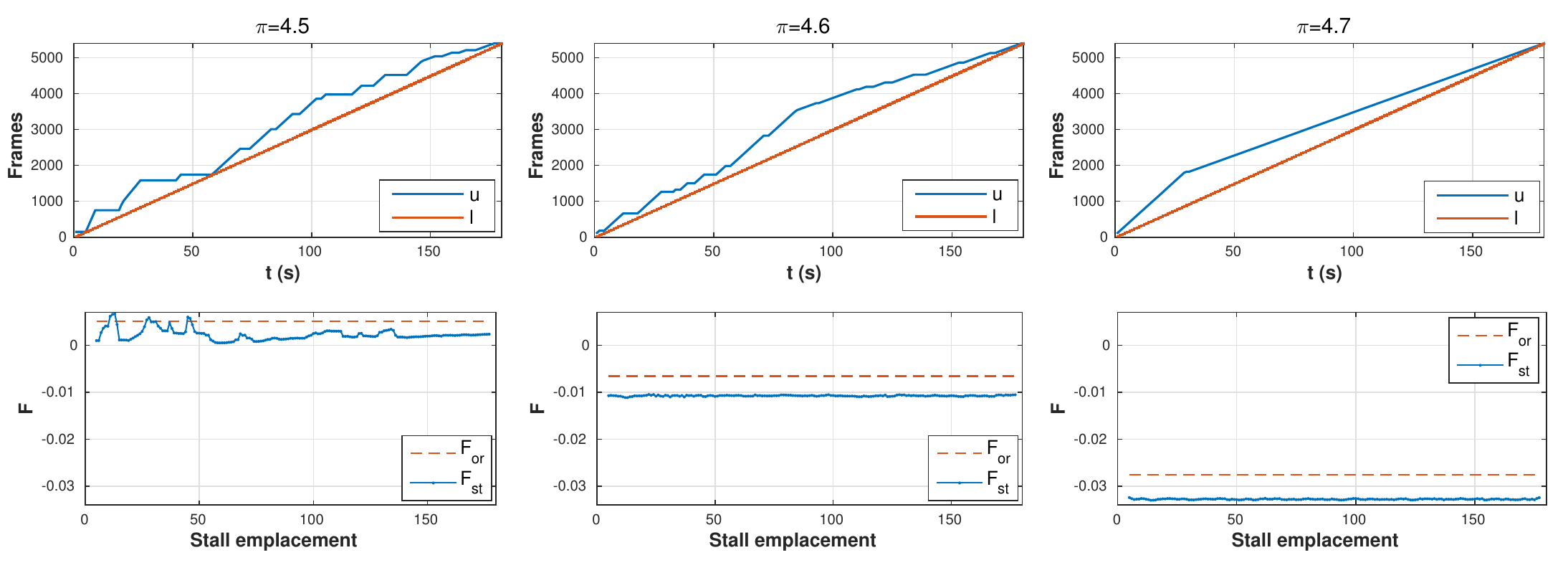}
	\vspace{-.2cm}
 \caption{System performance with buffer stall enforcement.}
	\label{fig:starvation}
\end{center}
\vspace{-.5cm}
 \end{figure*}

In Fig.~\ref{fig:constraints} we plot the playback buffer state evolution over time and its correspondent sequence of bitrates for the three aforementioned values of $\pi$. When $\pi$ is small, many silent times are noticed and the buffer state evolves with high slopes (mainly at the beginning and at the middle of the video). This is actually due to the low quality of the segments being streamed. Note that the player streams as much frames as the bitrate is low.
For the medium value of $\pi$, more flexibility is noticed with shorter silent times and better quality.
As for the big value of $\pi$, no silent times are noticed since almost all the network resources are used. The buffer state evolves gradually with low slopes, given the fact that segments of high-order quality are being streamed.

Now, we explore the idea of enforcing a stall during the streaming session. Let $\mathcal{F}_{or}$ be the original cost function before enforcing a stall, and $\mathcal{F}_{st}$ be the resulting cost function after enforcing a stall.
In Fig.~\ref{fig:starvation}, we plot again the playback buffer state evolution over time for the three values of $\pi$, and plot below the variation of $\mathcal{F}_{st}$ as function of the stall emplacement ($1^{st}$ segment, $2^{nd}$ segment, etc.). As depicted in the figure, for $\pi=4.5$, $\mathcal{F}_{st}$ experiences high fluctuations around $\mathcal{F}_{or}$ mainly when the stalls are enforced at the beginning of the video. The lowest values of $\mathcal{F}_{st}$ are noticed when the stalls are enforced at the moments where the original buffer state is critical, i.e., a low quality with no much flexibility toward the stall constraint. Note that, the critical states of the buffer at these moments prevent NEWCAST from setting a higher threshold. When a stall is enforced there, the video is divided into two independent parts and the streaming strategy is optimized before and after the stall, leading to two different thresholds that reduce the overall system utilization cost.
Now, by increasing $\pi$, we observe a quasi-constant decrease in $\mathcal{F}_{st}$. A stall enforcement certainly enhances the quality at the beginning part of the video, but it condemns the flexibility and the average quality for the rest of the video. The degradation in the global quality induces a reduction in the global system cost that outweighs the resulting $\mathcal{F}_{st}$. To sum it up, a stall enforcement may be only interesting when the value of $\pi$ is low since it may reduce the system cost. A judicious choice of its emplacement would be at the moments where the original buffer state is critical.

\subsection{Robustness under prediction errors}
One key limitation of the proposed idea is that there is still no explicit approach that accurately predicts the network capacity over more than ten seconds to the future. In order to evaluate the robustness of NEWCAST, we used the real throughput traces of the HSDPA dataset \cite{dataset}. This dataset consists of $30$ minutes of continuous throughput measurements of a moving device in Telenor's 3G/HSDPA wireless mobile network.
We used the traces of the Ljabru-Jernbanetorget trajectory as it has the least variance in the throughput spatial variation (see Fig.~\ref{capDistance} ). A temporal mapping of the throughput variation was performed by supposing the user moving at a speed of $50$Kmph. Using the same parameters of Table \ref{parameters}, we computed the performance $\mathcal{P}_{av}$ of NEWCAST by averaging all the throughput realizations, then, we computed its performance $\mathcal{P}_{real}$ by using each throughput realization apart.
The robustness of the framework was evaluated through the performance averaged error rate
$$\mathcal{P}_{error} = \left|\frac{\mathcal{P}_{real} - \mathcal{P}_{av}}{\mathcal{P}_{av}} \right|.$$
Results shown by Fig.~\ref{fig:Robustesse} depict an averaged error rate less than $15\%$ for both the system cost and the average quality. They even depict a lower sensitivity of the system cost to prediction errors when $\pi$ is smaller, and a lower sensitivity of the average quality to prediction errors when $\pi$ is higher. In general, we can claim that our scheme performs pretty well even with the presence of real prediction errors.

\begin{figure}[t]
\begin{center}
 \includegraphics[width=9.5cm,height=4.5cm]{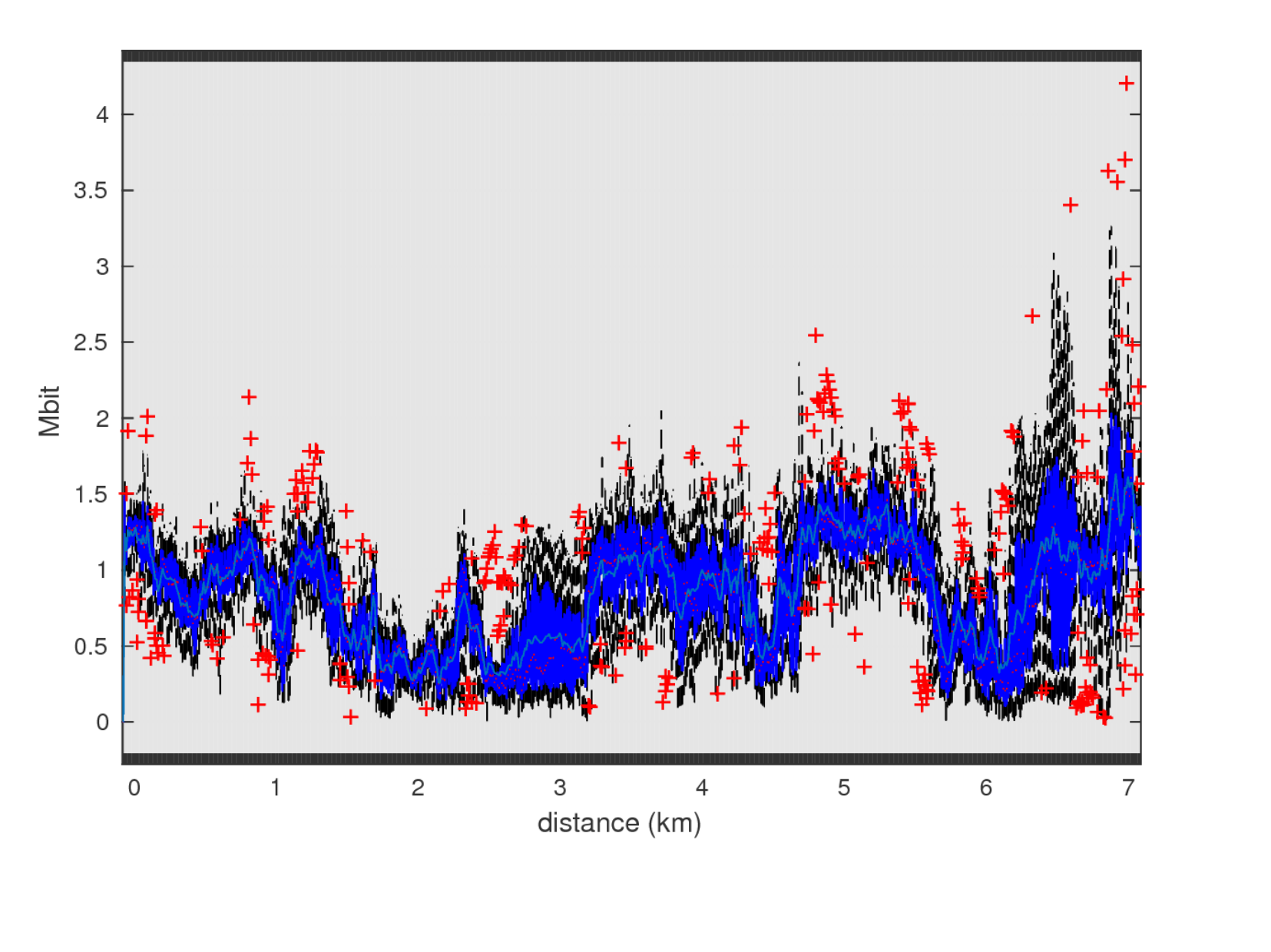}
 \vspace{-.7cm}
 \caption{Experimental spatial variations of the capacity for the tramway Ljabru-Jernbanetorget trajectory.}
 \label{capDistance}
\end{center}
\end{figure}

 \begin{figure}[t!]
 \begin{center}
 	\includegraphics[width=9cm] {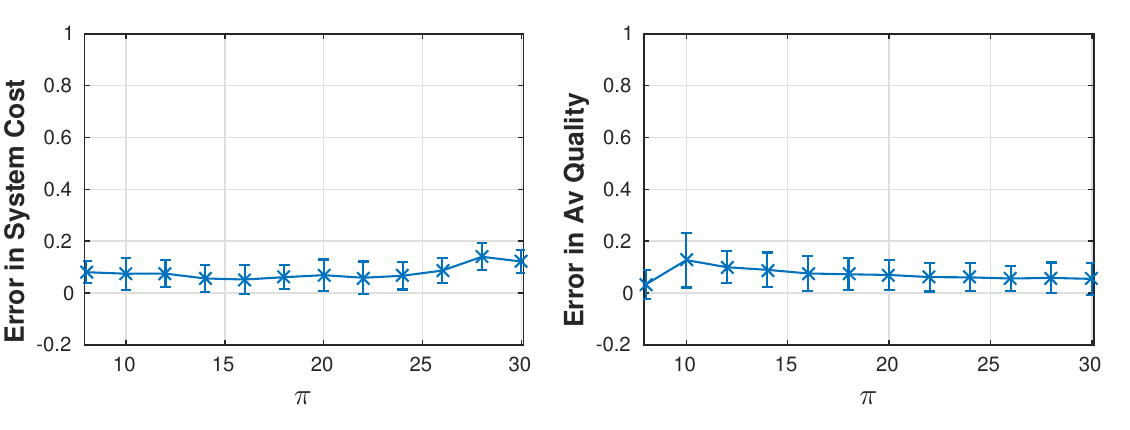}
	\vspace{-.7cm}
	\caption{Average error rate on the system performance under real throughput prediction errors.}
	\label{fig:Robustesse}
\end{center}
\end{figure}

\begin{figure}[t!]
 \begin{center}
 	\includegraphics[width=9cm]{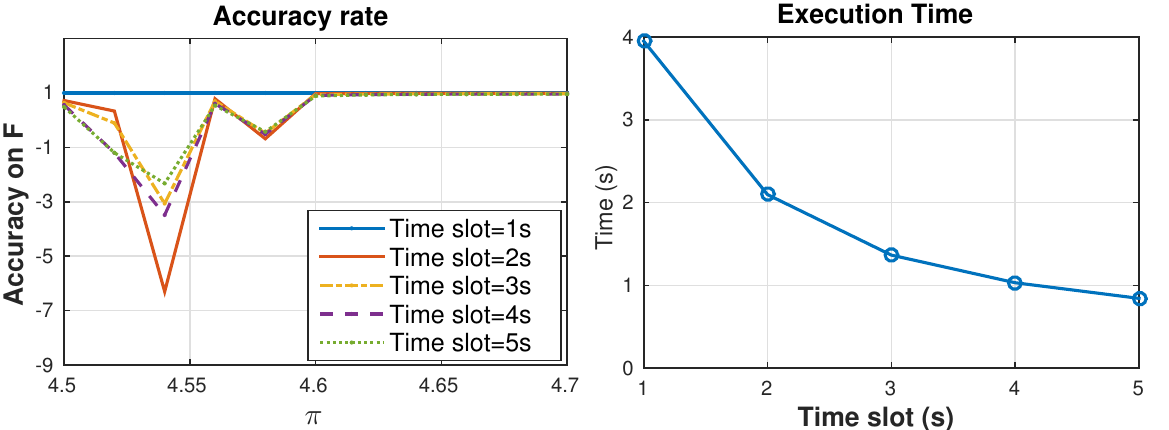}
	\vspace{-.5cm}
	\caption{Accuracy and complexity variations with different time slots.}
	\label{fig:AccuracyInterval}
 \end{center}
 \vspace{-.4cm}
 \end{figure}

 \begin{figure}[t!]
 \begin{center}
 	\includegraphics[width=9cm]{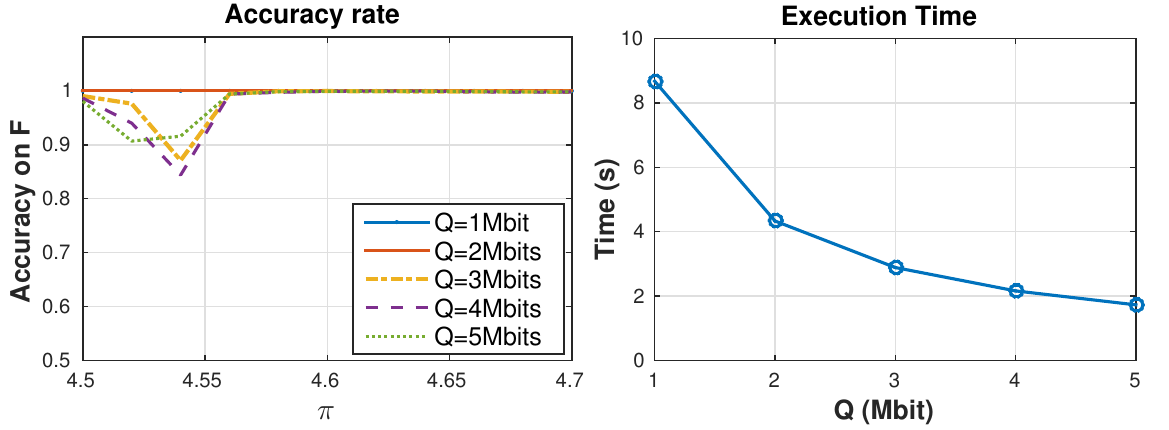}
	\vspace{-.5cm}
	\caption{Accuracy and complexity variations with different $Q$.}
	\label{fig:AccuracyQ}
 \end{center}
 \vspace{-.4cm}
 \end{figure}

\subsection{Complexity}
\label{complexity}
\subsubsection{Framework performance under bigger time slots.}

In Fig.~\ref{fig:AccuracyInterval}, we compute the mean execution time of NEWCAST (using optimal thresholds) by averaging results on $100$ (randomly generated) capacities and using different time slots (from $1$s to $5$s). It takes almost $4$s to compute the final strategy with a time slot equal to $1$s. As expected, using bigger time slots takes much shorter time. However, this comes at the expand of the final result accuracy depending on the value of $\pi$. In the same figure, we show the system response (through $\mathcal{F}$) for each time slot by averaging results over the $100$ capacities. We compute an accuracy rate factor ($\le 1$) by comparing the obtained results with the result of $1$s time slot. In our model, we assume that in a time slot only one bitrate level can be streamed, which explains why using bigger time slots may add constraints to the QoE. For high values of $\pi$, very slight degradation is noticed since the system tends to use all the network resources. However, for low values of $\pi$, the constraints have bigger impact since the system tends to use less network resources, which explains the higher degradation in the QoE, and, by the sequel, the higher reduction in the system cost. 

\subsubsection{Framework performance under different values of $Q$}

Here, we set the time slot to $1$s and run NEWCAST using different values of $Q$ (between $1$Mbit and $5$Mbits) by averaging results on the same $100$ capacities. Results in Fig.~\ref{fig:AccuracyQ} show that setting $Q$ to the average throughput ($2$Mbps) leads to a high accuracy rate ($\approx 1$) with an execution time of $4$s (as for optimal thresholds). Setting lower values of $Q$, increases the execution time and keeps almost the same accuracy on $\mathcal{F}$. For higher $Q$, the complexity is notably reduced, but slight degradations are noticed on the accuracy rate (less than $16\%$).
A judicious choice of $Q$ should then be made depending on the operator's preferences: a high $Q$ gives a high QoE and a very low complexity, whereas, a low $Q$ gives a low system cost and a higher complexity.

\subsection{Comparison with baseline adaptive bitrate (ABR) algorithms}
In this section, we compare NEWCAST to two baseline ABR algorithms: one is throughput-based (TB-ABR), the other is buffer-based (BB-ABR). We develop each algorithm on Matlab and simulate its behaviour on different video streaming sessions. We keep all the parameters setting of Tab \ref{parameters}.
\subsubsection{Overview on ABR algorithms}
The key difference between current ABR algorithms is the logic they use for bitrate selection. As found in the literature, ABR logics can be categorized in two main classes: throughput-based class \cite{Bouaziz, ACM12Lederer,Jiang:2012:IFE:2413176.2413189} and buffer-based class \cite{Huang14,DBLP:journals/jsac/ThangLPR14}. While the first class relies only on the next throughput prediction to decide on the current bitrate selection, the second class relies only on the current playback buffer occupancy. Few algorithms, however, were proposed as a mixture of throughput-based and buffer-based algorithms \cite{Yin:2015:CAD:2829988.2787486}.

 {\bf a. Throughput-based algorithms:} The commonly adopted rule in throughput-based algorithms is to never chose a bitrate larger than the estimated throughput unless it is the lowest bitrate. This is mainly to avoid eventual buffer underflows in the future. The key difference between these algorithms is the way to estimate the throughput and the way to select the bitrate.
According to \cite{ACM11Akhshabi}, throughput estimate is either \textit{instant} or \textit{smoothed}; The \textit{instant} throughput of a segment $i+1$ is the throughput measured during the download of segment $i$, whereas the \textit{smoothed} throughput is the weighted sum of the $n$ previous throughputs measured during the download of the latest $n$ segments, $n>1$ \cite{ICC15Parikshit}. The main drawback of using the \textit{instant} throughput is that it makes the bitrate selection react quickly to sudden throughput variations, which may annoy the user.
As for the bitrate selection, it can be either \textit{aggressive} or \textit{conservative} \cite{ACM12Huang}; With the \textit{aggressive} method, the bitrate can jump from one level to another without caring about the jump size, whereas the \textit{conservative} method tends to increase/decrease it progressively to not bother the user's perception. \\
{\bf b. Buffer-based algorithms:} The buffer-based method adapts the bitrate only based on the buffer occupancy. The buffer is in general divided into different ranges, and, depending on the range of the buffer level, multiple actions can be applied. In \cite{DBLP:journals/jsac/ThangLPR14}, a comparison between some existent buffer-based algorithms was done, according to authors, the \textit{most stable} one was addressed in \cite{ACM12Lederer} since the criterion to maintain the bitrate depends only on the range of the buffer level.

\subsubsection{TB-ABR and BB-ABR configuration: criteria of choice for comparison with NEWCAST}

The main characteristic of NEWCAST is that it increases the quality of segments {\textit{progressively}} to avoid bothering the user with sudden quality jumping. For this reason, we configure the TB-ABR and the BB-ABR algorithms to be both {\textit{conservative}}.
For TB-ABR we use \textit{the smoothed throughput} estimation such that

\begin{equation}\label{ThroughputEstim}
\hat{T}(i+1)=\sum_{k=i-3}^{i} p_k T(k),
\end{equation}
with $p_1=0.5; p_2=0.3; p_2=0.15; p_2=0.05$. $T(i)$ designs the throughput measured after downloading segment $i$, and $\hat{T}(i+1)$ is the throughput estimate of segment $i+1$. As for the bitrate selection, we use a method close to that defined in "Microsoft Smooth Streaming" (see Algorithm \ref{alg:TBA}).
For BB-ABR, we use the \textit{ most stable} method used in \cite{DBLP:journals/jsac/ThangLPR14}. We define three thresholds $B_{low} , B_{min}, $ and $B_{high}$ (respectively equal to $4$, $8$ and $12$ segments), and define multiple strategies of bitrate adaptation depending on the range of the buffer level (see Algorithm \ref{alg:BBA}).

 \begin{algorithm}
 \caption{TB-ABR : Throughput-Based ABR}
 \label{alg:TBA}
\SetAlgoLined\DontPrintSemicolon
\nl \For {segments of startup phase}{
\nl set the quality to the lowest bitrate $b_1$
}
\nl \For {segments of post-startup phase}{
\nl estimate the throughput based on the three previous downloaded segments \;
\nl \eIf {the throughput $\le b_1$}{
set the quality to $b_1$
}{
 \nl \eIf {the throughput $\le$ the previous bitrate }{
 \nl set the quality to the highest bitrate below the throughput \;
 }{
 \nl \eIf {the next higher bitrate $\le$ the throughput}{
 increase the bitrate by one level
 }{
 keep the same quality}
 }
}
}
 \vspace{.5cm}
\end{algorithm}

 \begin{algorithm}
 \caption{BB-ABR : Buffer-Based ABR}
 \label{alg:BBA}
\SetAlgoLined\DontPrintSemicolon
\nl \For {segments of startup phase}{
\nl set the quality to the lowest bitrate $b_1$
}
\nl \For {segments of post-startup phase}{
\nl \If { $0 \le$ BufferState $\le B_{min}$}{
\nl set the quality to the lowest bitrate $b_1$
}
\nl \If { $B_{min} <$ BufferState $\le B_{low}$}{
\nl \eIf {the BufferState is increasing }{
\nl keep the same quality
}{
\nl decrease the bitrate by one level if possible, or keep it the same
}
}

\nl \If {$B_{low}<$ BufferState $\le B_{high} $}{
\nl keep the same quality}

\nl \If {$B_{high} >$ BufferState}{
\nl increase the bitrate by one level if possible or keep it the same}
}
 \vspace{.5cm}
\end{algorithm}

\subsubsection{Capacities of test}
To be as close as possible to real world throughput variations, we generate the capacities of test by the mean of the standard-complaint Ns3 simulator. We conducted extensive simulations of an LTE-network by varying the mobility or/and the number of users each time. All the throughput samples resulting from these simulations were used for the evaluation of both NEWCAST and the ABR algorithms. 


\subsubsection{Main comparison points}

 \begin{figure}[h]
 \begin{center}
 	\includegraphics [width=9.5cm] {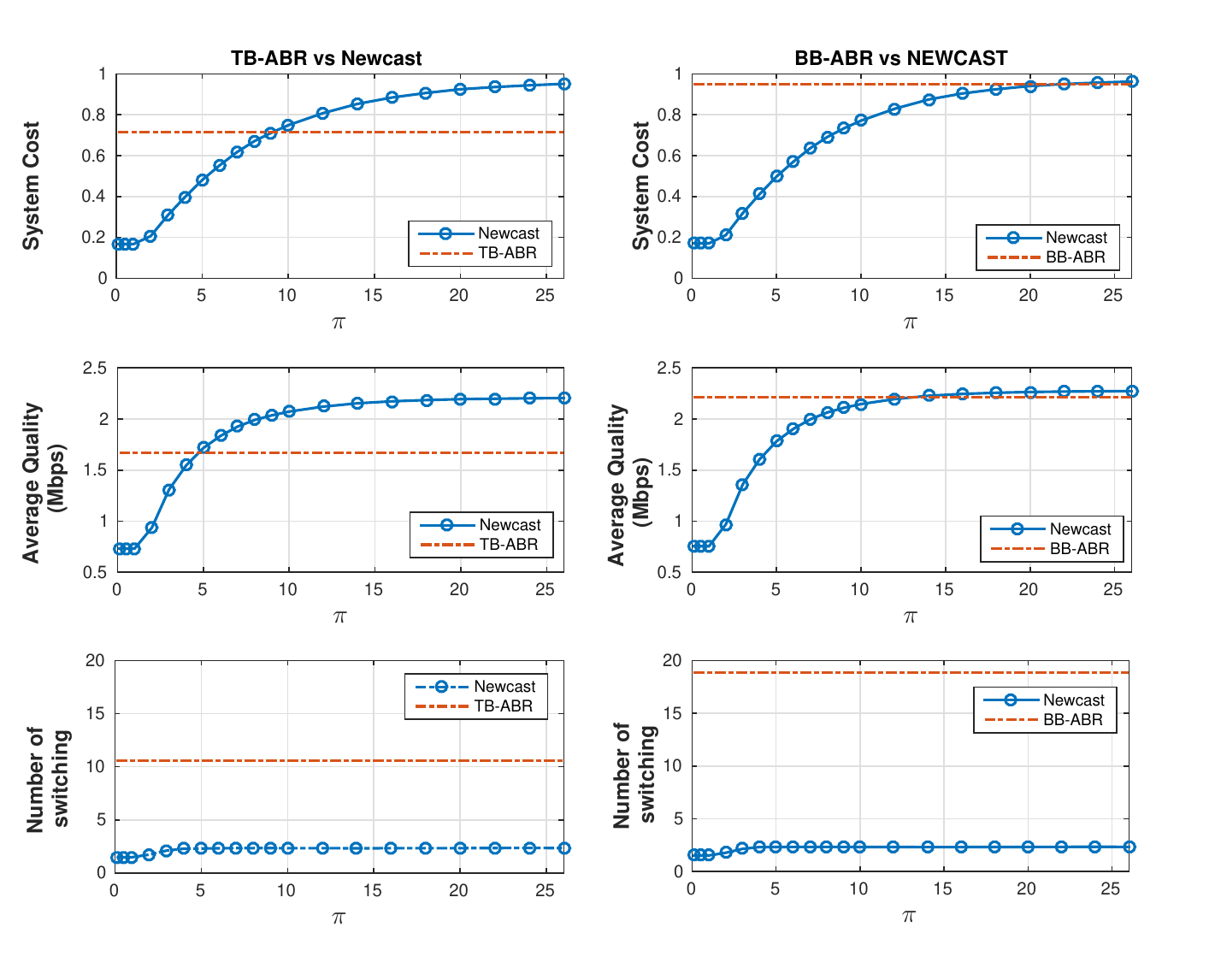}
 \vspace{-.7cm}
	\caption{TB-ABR vs. NEWCAST and BB-ABR vs. NEWCAST (without stalls).}
	\label{fig:DASHvsNewcastNs3}
	\end{center}
 \end{figure}

After having run NEWCAST and the two aforementioned ABR algorithms using all the throughput samples, we noticed that the TB-ABR algorithm was encountering video stalls (at least one stall) in $0.3\%$ of cases, whereas the BB-ABR algorithm was encountering video stalls in $7.8\%$ of cases. NEWCAST, however, succeeded at achieving zero stall during all the streaming sessions. Hence, we found it more judicious to perform the comparison by distinguishing the cases where the number of stalls encountered by each ABR algorithm was also equal to zero. Our analysis is driven by the three metrics that mostly characterize NEWCAST: the system cost, the average per segment video quality, and the average number of quality switching. In Fig.\ref{fig:DASHvsNewcastNs3}, we plot each of these metrics as function of $\pi$; $\pi$ ranging from $0.2$ to $26$.

{\bf a. TB-ABR vs. NEWCAST}: According to Fig.~\ref{fig:DASHvsNewcastNs3}, the main advantage of NEWCAST is that it can achieve the same quality as TB-ABR with a system utilization cost reduced by at least $30\%$, and that it can achieve the same system cost with an average quality enhanced by up to $19\%$. This is mainly due to the smart threshold-based-strategy of NEWCAST that uses the less expensive resources depending on the value of $\pi$. It is then up to the operator to make the tradeoff and to wisely calibrate the value of $\pi$ to outperform the TB-ABR algorithm. A further important observation lies in the very reduced number of quality switching achieved by NEWCAST (at most $2.5$) compared to that achieved by TB-ABR (around $11$). \\
 {\bf b. BB-ABR vs. NEWCAST}: We notice from Fig.~\ref{fig:DASHvsNewcastNs3} that BB-ABR is very greedy toward the resource usage compared to TB-ABR, which makes it give near performance to NEWCAST when applied with high values of $\pi$. Actually, for some values of $\pi$, NEWCAST outperforms BB-ABR, but this outperformance is marginal. In fact, the same average quality can be achieved with a system cost reduced by $12\%$, and the same system cost can be achieved resulting in an average quality increased by only $4\%$. The greedy character of BB-ABR can be either emphasised or de-emphasised depending on the thresholds set for the playback buffer ($B_{min}$, $B_{low}$ and $B_{high}$), so it may happen that BB-ABR uses all the resources and gives a higher average quality than NEWCAST, but this outperformance will not exceed $2\%$ since the heuristic used by NEWCAST approximates the optimal quality arrangement by $98\%$. All things considered, the most worth citing advantage of NEWCAST, is that it gives a far less number of quality switching (at most $2.5$ against $19$ with BB-ABR), which is quite better for the users' perception.\\

In conclusion, when the knowledge of the future throughput is perfect, NEWCAST can perform better than the baseline TB-ABR and BB-ABR algorithms. By mean of a wise calibration of the value of $\pi$, the tradeoff between system utilization cost and QoE can be steered to either save more resources or increase the average quality. In all cases, the number of quality switching remains the most suitable for the end user's perception.

\section{Framework design and implementation} \label{sec:experiments}
\subsection{NEWCAST interactions with real video streaming entities}

In real environments, NEWCAST must be implemented at the client side as an independent framework.
It must be able to communicate the threshold $\alpha_{th}$ to the network scheduler and the set of video bitrates $\gamma_{th}$ to the media player as described in Fig.\ref{fig:framework}.
We can imagine for sending $\alpha_{th}$ a kind of a cross layer that also allows to apply the threshold-based transmission scheme. The set of video bitrates $\gamma_{th}$, however, can be directly sent to the player just at the beginning of the streaming session. These bitrates will then be consecutively requested by the player to the streaming server. Note that, in our analytical model, the variable $\gamma_{th}$ was set to describe the variation of the video bitrate \textit{in function of time}, in real implementation, the player will not use it that way, it will rather use the bitrate variation \textit{in function of the segments' orders}, which can be directly returned by NEWCAST. In Fig.~\ref{fig:interactions} we show the sequence diagram of the video streaming process using NEWCAST.

 \begin{figure}[t]
 \centering
 	\includegraphics [width=6cm] {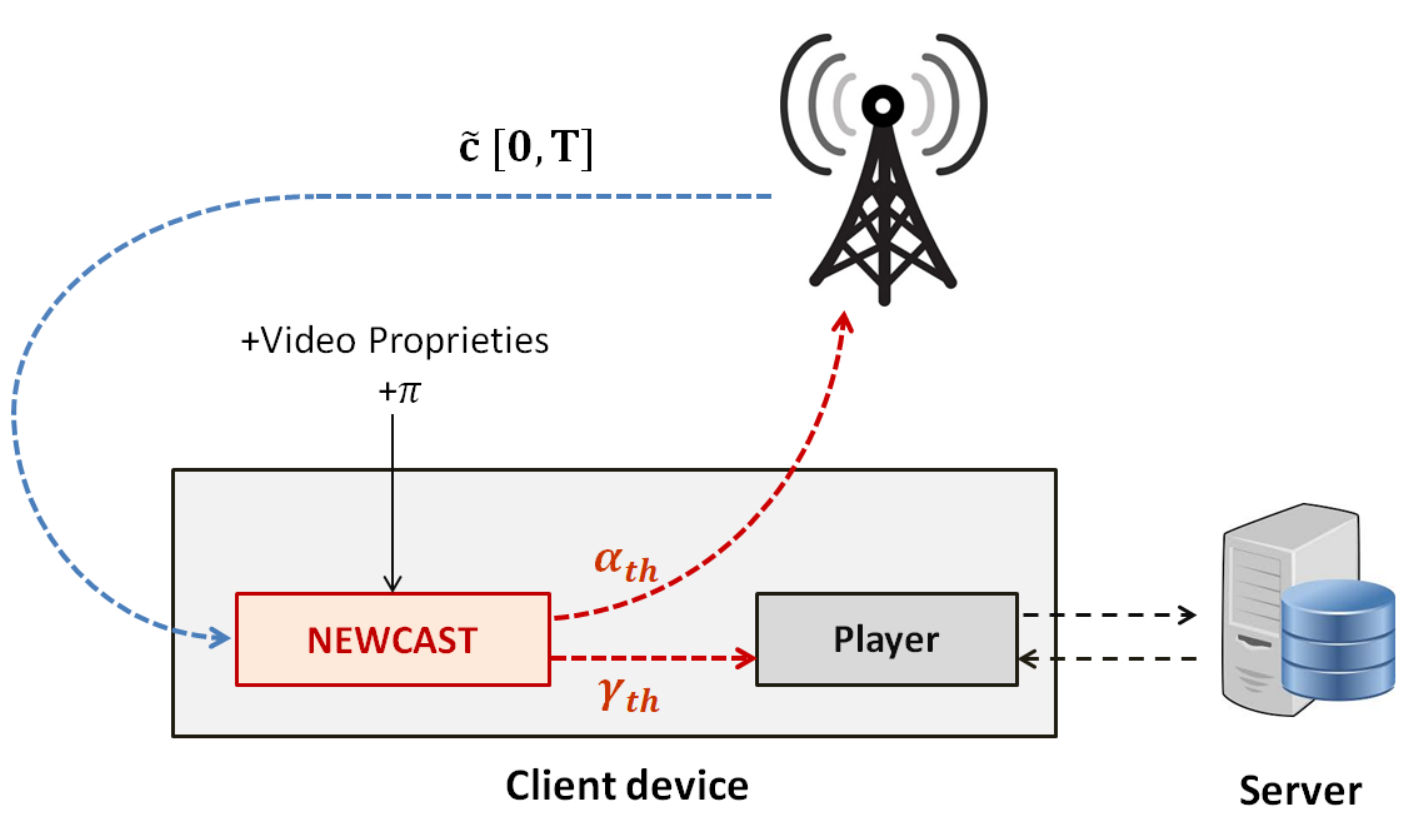}
	\caption{Illustration of NEWCAST interactions with the network scheduler and the media player.}
	\label{fig:framework}
 \end{figure}

 \begin{figure}[t]
 \begin{center}
 	\includegraphics[width=9.5cm]{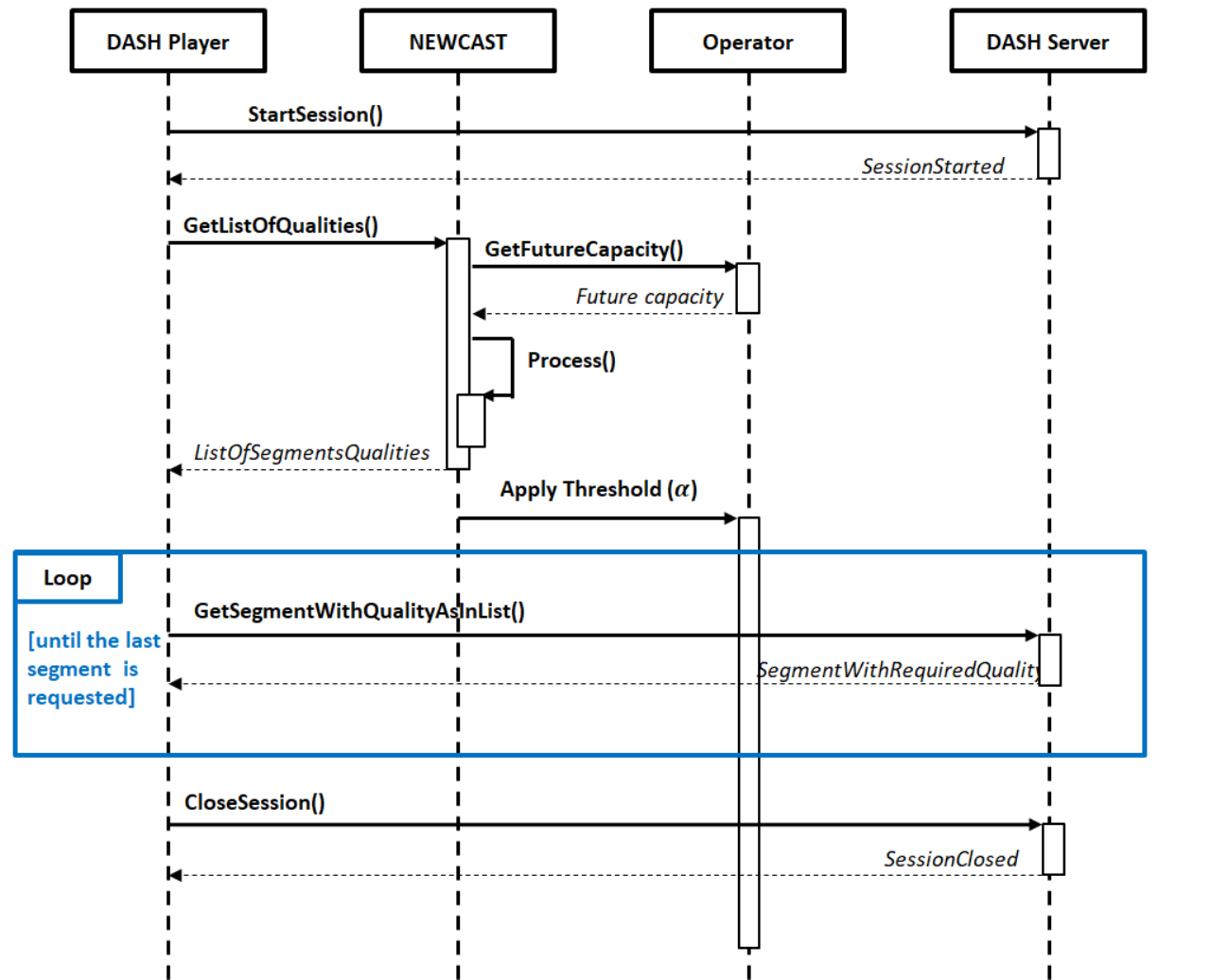}
 \vspace{-.1cm}
	\caption{Sequence diagram of a video streaming session using NEWCAST.}
	\label{fig:interactions}
 \end{center}
 \vspace{-0.3cm}
 \end{figure}

\subsection{Implementation tools and environment}

We use a Linux environment with two virtual machines: one is used as a DASH server and the other is used as a DASH client. In the DASH server, we install Apache and put inside the Dashjs framework \cite{dashjs} with the Envivio video segments encoded at different quality levels \cite{envivio}. In the DASH client, we only install Google chrome browser. We configure the two virtual machines to be able to communicate through their Ethernet interfaces. To emulate the network schedule and make the bandwidth between the two machines follow a predefined variation (considered as the predicted capacity), we use the Linux tc-tool for traffic shaping as shown in Fig.~\ref{fig:RealSystem}. To develop NEWCAST and make it interact with the Dashjs player, we use Javascript and other basic web languages. NEWCAST is put with the player call function in a same \textit{.php} file that the DASH client requests to start the video streaming session. A video demo of NEWCAST is put available online in \cite{demoNEWCAST}. In Table~\ref{detailsImplementation}, we put more details on the hardware/software tools used for the implementation.

 \begin{figure}[t]
 \begin{center}
 	\includegraphics[width=8.5cm]{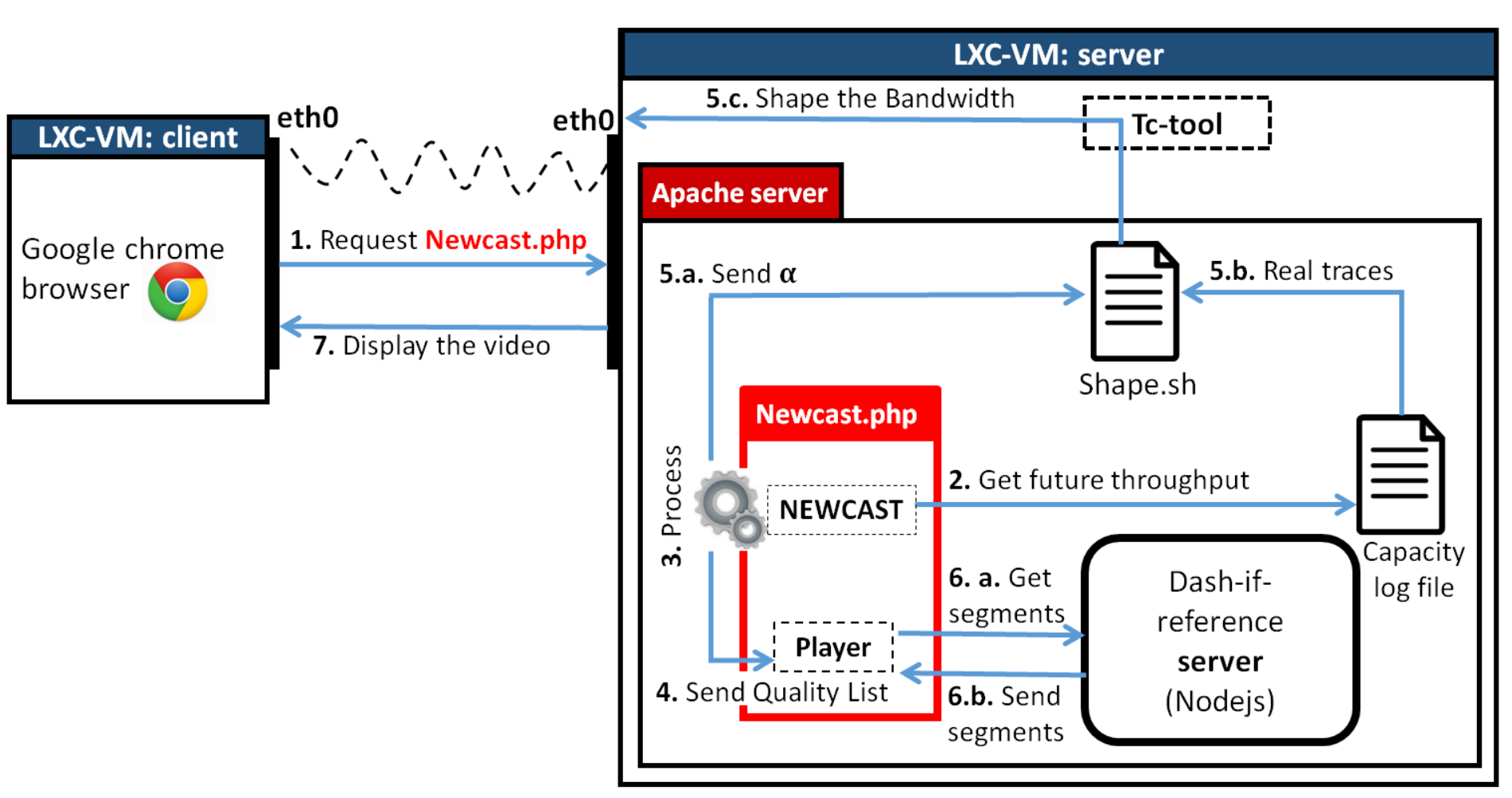}
 \vspace{-.1cm}
	\caption{Architecture of the system used for experiments.}
	\label{fig:RealSystem}
 \end{center}
 \vspace{-.3cm}
 \end{figure}

\begin{table}
\begin{center}
\begin{tabular}{|c|c|}
\hline
\textbf{Host machine} & Optiplex 7010 Intel Core i7-3770 CPU 3.40Ghz \\
\hline
\textbf{Distribution} & Ubuntu 14.04.5 LTS \\
\hline
\textbf{Virtual machines} & Linux Container Lxc 1.0.9 \\
\hline
\textbf{Apache} & 2.4.7 \\
\hline
\textbf{Dashjs} & 2.4.0 \\
\hline
\textbf{Google Chrome} & 55.0.2883.87 \\
\hline
\end{tabular}
	\end{center}
	\caption{Details on the software/hardware tools used for real implementation.}
	\label{detailsImplementation}
 \vspace{-0.5cm}
\end{table}

\subsection{Requirements for real implementation }
\subsubsection{Changes inside the Dashjs framework}
To make NEWCAST interact with the Dashjs framework, we made some changes inside the media player:
(i) A new event was added to the player class to detect the moments where a segment of type \textit{"video"} is completely loaded to the client.
(ii) The restrictions on the playback buffer size defined at the \textit{"MediaPlayerModel.js"} file were changed to fit the infinite buffer size assumption, since, otherwise, the player will remove the earliest played segments and, in some cases, delay the requests of the coming segments.
(iii) The threshold of prefetching \textit{after a stall happens} was changed inside the \textit{checkIfSufficientBuffer()} function to fit the prefetching threshold used by NEWCAST.
\subsubsection{Required player APIs}
Two essential APIs are actually responsible for the interaction between NEWCAST and the Dashjs framework: The \textit{setAutoSwitchQuality()} API to disable the quality auto-switch mode of the player and the \textit{setQualityFor()} API to enforce the quality of the coming segments.

\subsubsection{Traffic configuration}

To make the \textit{real} throughput compliant to the throughput $r$ modelled by NEWCAST, we processed as follows:
(i) We deleted the \textit{"audio traffic"} description from the \textit{.mpd} file since, in our study, we are only interested in \textit{video traffic}.
(ii) We added Apache to the Linux sudoer list to allow it use the tc-tool functions and shape the bandwidth in parallel to the streaming.
(iii) As we found that the average duration of a real \textit{segment-request} is equal to $0.06$s, which is not insignificant as was assumed in our theoretical model, we were considering, for the implementation, each \textit{segment-request} as a virtual file of size $0.06$ multiplied by the predicted throughput at the considered second.
(iv) A long stall duration caused by a high threshold $\alpha_{th}$ may lead the session to be closed. To avoid such situations we were disabling the threshold transmission schedule during prefetching when a stall happens.
	
\subsection{Validation through experiments}

\begin{figure}[t]
 \begin{center}
 	\includegraphics[width=9cm]{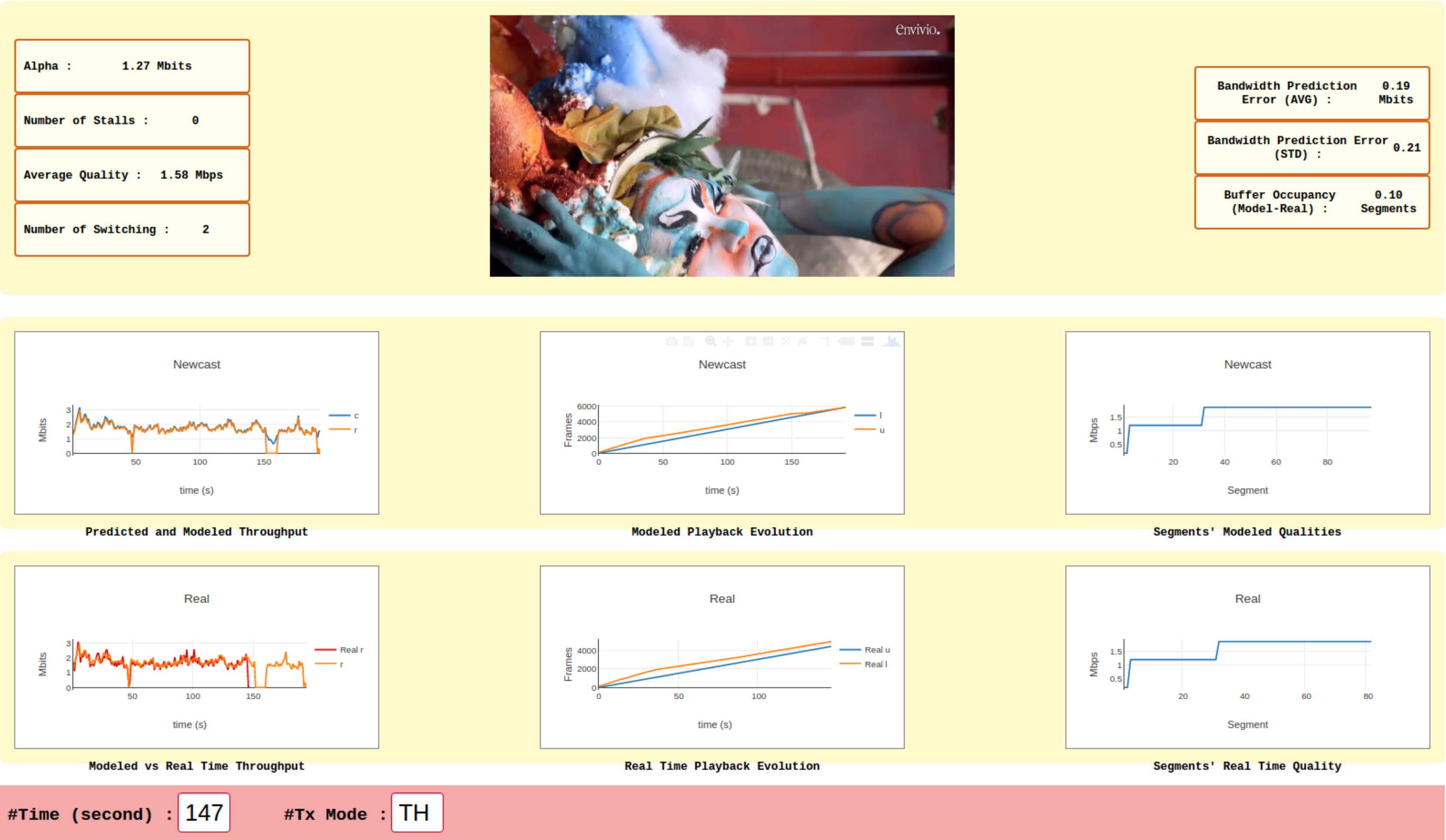}
 \vspace{-.2cm}
	\caption{Screenshot of the graphical interface.}
	\label{fig:ScreenShotNewcast}
 \end{center}
 \vspace{-.5cm}
 \end{figure}

\begin{figure}[t]
 \begin{center}
 	\includegraphics[width=8cm]{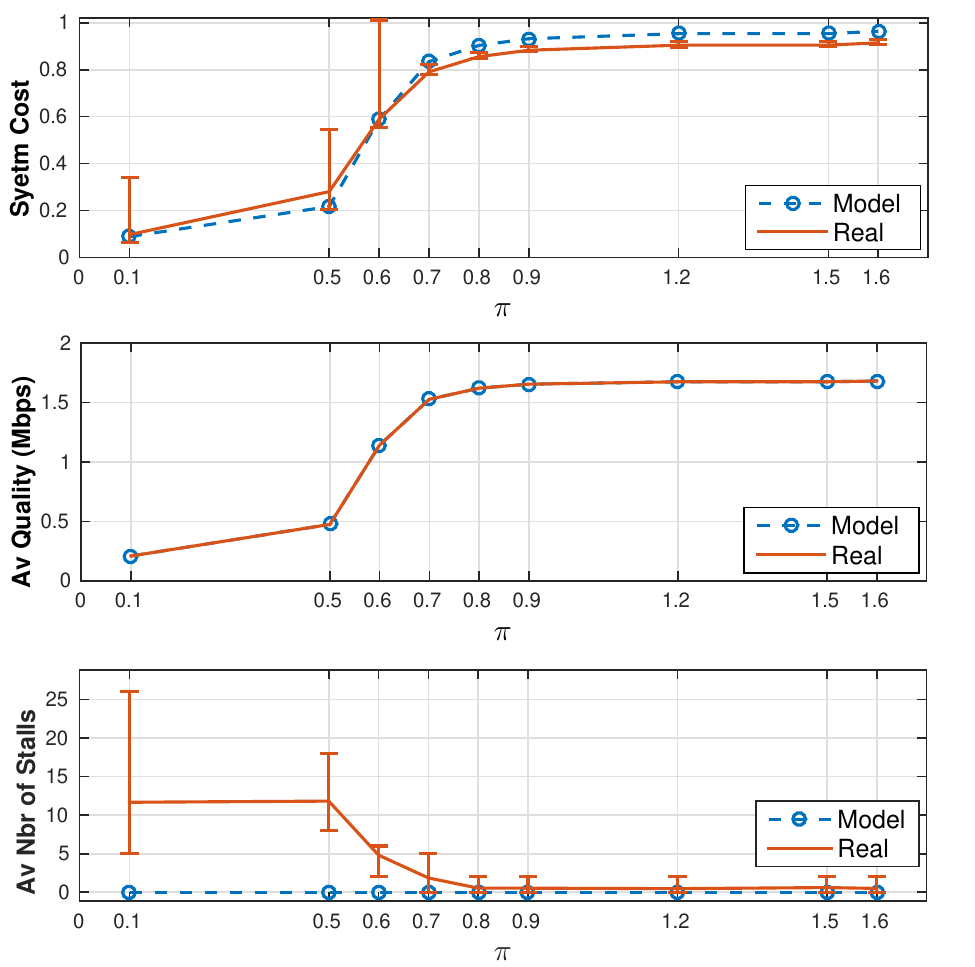}
 \vspace{-.2cm}
	\caption{NEWCAST performance in real environment.}
	\label{fig:NewcastWithTctool}
 \end{center}
 \vspace{-.5cm}
 \end{figure}

To supervise the system behaviour in real time, we developed a graphical interface in which we plot the real time throughput variation, the real time buffer evolution and the real time video bitrate alongside with the strategies modeled by NEWCAST, as shown in the screenshot of Fig.\ref{fig:ScreenShotNewcast} and in our video demo \cite{demoNEWCAST}. We conducted the same experiment several times using one of the throughput logs available in \cite{dataset} and different values of $\pi$. Results shown by Fig.~\ref{fig:NewcastWithTctool} depict a high instability of the system behaviour when the value of $\pi$ is small (between $0.1$ and $0.7$), i.e., when the threshold $\alpha_{\pi}$ is high. This instability actually induced high numbers of video stalls. More stability, however, is noticed when the value of $\pi$ is high (between $0.8$ and $1.6$). A noteworthy observation here is in the fact that the real system reacts very closely to what was modeled by NEWCAST. The difference in the system utilization cost is very small (approximately equal to $5.2\%$) and the average number of video stalls too (almost equal to $0.53$). Although we were conducting the same experiments for each value of $\pi$, the system behaviour was variable. We link this, mainly, to the casual errors of the bandwidth shaping and to the variation of the \textit{segment-request} duration. Overall, these results offer hope that, under high values of $\pi$, the exploitation of NEWCAST in real environments becomes feasible, unless an accurate throughput prediction is available. Under low values of $\pi$, however, the quality of the streaming risks to be degraded since the system becomes sensitive to the tiniest prediction error.

\section{Conclusion}\label{sec:conc}
In this paper, we have developed a new framework called NEWCAST, for optimizing the delivery of video streaming content under the knowledge of future capacity. This framework has been designed to balance the system utilization cost and some key QoE metrics such as average video quality and rebuffering events. From an implementation point of view, results have shown the possibility to use NEWCAST as an online algorithm (well suited for dynamic adaptive streaming over HTTP).
Real experiments conducted with a real DASH player have shown that NEWCAST can be efficiently used in real-world streaming provided that the throughput is accurately estimated. Interesting future directions consist in incorporating errors in the throughput prediction to see how to improve the robustness of our approach.

\section*{Acknowledgement}
This work has been carried out in the framework of the ANR IDEFIX project, funded by the ANR under the contract number ANR-13-INFR-0006.
\bibliographystyle{IEEEtran}
\bibliography{ImenBib}

\begin{thebibliography}{10}
\providecommand{\url}[1]{#1}
\csname url@samestyle\endcsname
\providecommand{\newblock}{\relax}
\providecommand{\bibinfo}[2]{#2}
\providecommand{\BIBentrySTDinterwordspacing}{\spaceskip=0pt\relax}
\providecommand{\BIBentryALTinterwordstretchfactor}{4}
\providecommand{\BIBentryALTinterwordspacing}{\spaceskip=\fontdimen2\font plus
\BIBentryALTinterwordstretchfactor\fontdimen3\font minus
  \fontdimen4\font\relax}
\providecommand{\BIBforeignlanguage}[2]{{%
\expandafter\ifx\csname l@#1\endcsname\relax
\typeout{** WARNING: IEEEtran.bst: No hyphenation pattern has been}%
\typeout{** loaded for the language `#1'. Using the pattern for}%
\typeout{** the default language instead.}%
\else
\language=\csname l@#1\endcsname
\fi
#2}}
\providecommand{\BIBdecl}{\relax}
\BIBdecl

\bibitem{dataset}
{DATASET: HSDPA-bandwidth logs for mobile HTTP streaming scenarios}.
  \url{http://home.ifi.uio.no/paalh/dataset/hsdpa-tcp-logs/}.

\bibitem{Cisco16}
\emph{{Cisco Visual Networking Index: Forecast and Methodology, 2016-2021}},
  \url{https://www.cisco.com/c/en/us/solutions/collateral/service-provider/visual-networking-index-vni/complete-white-paper-c11-481360.pdf}.

\bibitem{conf/im/VriendtVR13}
J.~De~Vriendt, D.~De~Vleeschauwer, and D.~Robinson, ``{Model for estimating QoE
  of video delivered using HTTP adaptive streaming},'' in \emph{Integrated
  Network Management (IM 2013), 2013 IFIP/IEEE International Symposium on}, May
  2013, pp. 1288--1293.

\bibitem{DOB11}
{ F. Dobrian, A. Awan, D. Joseph, A. Ganjam, J. Zhan, V. Sekar, I. Stoica, and
  H. Zhang}, ``{Understanding the impact of video quality on user
  engagement},'' in \emph{ACM SIGCOMM}, 2011.

\bibitem{Sigcomm13}
A.~Balachandran, V.~Sekar, A.~Akella, S.~Seshan, I.~Stoica, and H.~Zhang,
  ``Developing a predictive model of quality of experience for internet
  video,'' in \emph{ACM SIGCOMM}, 2013.

\bibitem{Yim11}
C.~Yim and A.~C. Bovik, ``Evaluation of temporal variation of video quality in
  packet loss networks,'' in \emph{Signal Processing: Image Communication},
  2011.

\bibitem{quality}
A.~K. Moorthy, L.~K. Choi, A.~C. Bovik, and G.~de~Veciana, ``Video quality
  assessment on mobile devices: Subjective, behavioral and objective studies,''
  \emph{J. Sel. Topics Signal Processing}, pp. 652--671, 2012.

\bibitem{Stockhammer:2011:DAS:1943552.1943572}
T.~Stockhammer, ``{Dynamic Adaptive Streaming over HTTP: Standards and Design
  Principles},'' in \emph{Second Annual ACM Conference on Multimedia Systems},
  2011.

\bibitem{Huang14}
T.~Y. Huang, R.~Johari, N.~McKeown, M.~Trunnell, and M.~Watson, ``{A
  buffer-based approach to rate adaptation: Evidence from a large video
  streaming service},'' in \emph{ACM Conference on Special Interest Group on
  Data Communication (SIGCOMM)}, 2014.

\bibitem{Tian:2012:TAS:2413176.2413190}
G.~Tian and Y.~Liu, ``Towards agile and smooth video adaptation in dynamic http
  streaming,'' in \emph{Proceedings of the 8th International Conference on
  Emerging Networking Experiments and Technologies}, 2012.

\bibitem{Yin:2015:CAD:2829988.2787486}
X.~Yin, A.~Jindal, V.~Sekar, and B.~Sinopoli, ``{A Control-Theoretic Approach
  for Dynamic Adaptive Video Streaming over HTTP},'' \emph{SIGCOMM Comput.
  Commun. Rev.}, pp. 325--338, 2015.

\bibitem{TCPthroughput}
{A. Jain and A. Terzis and N. Sprecher and P. Szilagyi and H. Flinck},
  ``{Mobile Throughput Guidance Signaling Protocol
  draft-flinck-mobile-throughput-guidance-00},'' \emph{IETF}, April 2014.

\bibitem{Bouaziz}
C.~Liu, I.~Bouazizi, and M.~Gabbouj, ``{Rate adaptation for adaptive HTTP
  streaming},'' in \emph{ACM Multimedia Syst.}, 2011.

\bibitem{Jiang:2012:IFE:2413176.2413189}
J.~Jiang, V.~Sekar, and H.~Zhang, ``{Improving Fairness, Efficiency, and
  Stability in HTTP-based Adaptive Video Streaming with FESTIVE},'' in
  \emph{Proceedings of the 8th International Conference on Emerging Networking
  Experiments and Technologies}, 2012.

\bibitem{journals/corr/JosephV13}
V.~Joseph and G.~de~Veciana, ``{NOVA: QoE-driven optimization of DASH-based
  video delivery in networks},'' in \emph{INFOCOM, Proceedings IEEE}, April
  2014.

\bibitem{conf/infocom/LuV13}
Z.~Lu and G.~de~Veciana, ``Optimizing stored video delivery for mobile
  networks: The value of knowing the future,'' in \emph{INFOCOM}, 2013.

\bibitem{5G}
{Expert Working Group on 5G Challenges Research Priorities, and
  Recommendations}, ``Networld 2020 etp,'' \emph{White paper}, Aug. 2014.

\bibitem{Zou:2015:API}
X.~K. Zou, J.~Erman, V.~Gopalakrishnan, E.~Halepovic, R.~Jana, X.~Jin,
  J.~Rexford, and R.~K. Sinha, ``Can accurate predictions improve video
  streaming in cellular networks?'' in \emph{Proceedings of the 16th
  International Workshop on Mobile Computing Systems and Applications}, 2015.

\bibitem{YoutubeResolutions}
\emph{Live encoder settings, bitrates and resolutions},
  \url{https://support.google.com/youtube/answer/2853702?hl=en}, consulted in
  2015-07-18.

\bibitem{Lederer:2012:DAS:2155555.2155570}
S.~Lederer, C.~M\"{u}ller, and C.~Timmerer, ``{Dynamic Adaptive Streaming over
  HTTP Dataset},'' in \emph{Proceedings of the 3rd Multimedia Systems
  Conference}, 2012.

\bibitem{DBLP:conf/icc/ShenLLY14}
Y.~Shen, Y.~Liu, Q.~Liu, and D.~Yang, ``{A method of QoE evaluation for
  adaptive streaming based on bitrate distribution},'' in \emph{IEEE
  International Conference on Communications, Workshops Proceedings}, 2014.

\bibitem{conf/icc/EssailiSSSKS13}
A.~E. Essaili, D.~Schroeder, D.~Staehle, M.~Shehada, W.~Kellerer, and E.~G.
  Steinbach, ``{Quality-of-experience driven adaptive HTTP media
  delivery}.''\hskip 1em plus 0.5em minus 0.4em\relax IEEE, 2013, pp.
  2480--2485.

\bibitem{ACM12Lederer}
S.~Lederer, C.~M\"{u}ller, and C.~Timmerer, ``Dynamic adaptive streaming over
  http dataset,'' in \emph{Proceedings of the 3rd Multimedia Systems
  Conference}.\hskip 1em plus 0.5em minus 0.4em\relax ACM, 2012.

\bibitem{DBLP:journals/jsac/ThangLPR14}
T.~C. Thang, H.~T. Le, A.~T. Pham, and Y.~M. Ro, ``{An Evaluation of Bitrate
  Adaptation Methods for {HTTP} Live Streaming},'' \emph{{IEEE} Journal on
  Selected Areas in Communications}, 2014.

\bibitem{ACM11Akhshabi}
S.~Akhshabi, A.~C. Begen, and C.~Dovrolis, ``An experimental evaluation of
  rate-adaptation algorithms in adaptive streaming over http,'' in
  \emph{Proceedings of the Second Annual ACM Conference on Multimedia
  Systems}.\hskip 1em plus 0.5em minus 0.4em\relax ACM, 2011.

\bibitem{ICC15Parikshit}
P.~Juluri, V.~Tamarapalli, and D.~Medhi, ``{SARA:} segment aware rate
  adaptation algorithm for dynamic adaptive streaming over {HTTP},'' in
  \emph{{IEEE} International Conference on Communication, {ICC} 2015, London,
  United Kingdom, June 8-12, 2015, Workshop Proceedings}, 2015.

\bibitem{ACM12Huang}
T.-Y. Huang, N.~Handigol, B.~Heller, N.~McKeown, and R.~Johari, ``Confused,
  timid, and unstable: Picking a video streaming rate is hard,'' in
  \emph{Proceedings of the 2012 Internet Measurement Conference}.\hskip 1em
  plus 0.5em minus 0.4em\relax ACM, 2012.

\bibitem{dashjs}
\emph{{Dash-Industry-Forum}},
  \url{https://github.com/Dash-Industry-Forum/dash.js/}.

\bibitem{envivio}
\emph{{I}ndex of /129021/dash/envivio/Envivio-dash2},
  \url{http://dash.edgesuite.net/envivio/Envivio-dash2}.

\bibitem{demoNEWCAST}
{Video demo of NEWCAST}.
  \url{https://drive.google.com/open?id=0B1gjdIZb5PPIcW1OLWY4d2xKS2s}.

\end{thebibliography}

\end{document}